\documentclass[a4paper,UKenglish,cleveref, autoref, thm-restate]{lipics-v2021}

\nolinenumbers

\usepackage{tikz}
\usetikzlibrary{patterns}
\usepackage{bbm}
\usepackage{amsmath,amsfonts} 
\usepackage{algorithmic,algorithm}
\usepackage{graphicx}
\usepackage{caption}
\usepackage{subcaption}
\usepackage{textcomp}
\usepackage{xcolor}

\newcommand{\av}[1]{\overline{#1}}
\newcommand{\aaoi}{\overline{\Delta}}

\title{Minimizing Age of Information under Arbitrary Arrival Model with Arbitrary Packet Size} 

\titlerunning{Minimizing Age of Information under Arbitrary Arrival Model with Arbitrary Packet Size} 

\author{Kumar Saurav}{Tata Institute of Fundamental Research, Mumbai, India 
}
{kumar.saurav@tifr.res.in}{https://orcid.org/0000-0001-8381-1228}{}

\author{Rahul Vaze
}{Tata Institute of Fundamental Research, Mumbai, India}{rahul.vaze@gmail.com}{https://orcid.org/0000-0002-9712-1110}{}

\authorrunning{K. Saurav and R. Vaze} 

\Copyright{Kumar Saurav and Rahul Vaze} 



\relatedversion{} 

\funding{We acknowledge support of the Department of Atomic Energy, Government of India, under project no. RTI4001.}

\hideLIPIcs  

\EventEditors{John Q. Open and Joan R. Access}
\EventNoEds{2}
\EventLongTitle{42nd Conference on Very Important Topics (CVIT 2016)}
\EventShortTitle{CVIT 2016}
\EventAcronym{CVIT}
\EventYear{2016}
\EventDate{December 24--27, 2016}
\EventLocation{Little Whinging, United Kingdom}
\EventLogo{}
\SeriesVolume{42}
\ArticleNo{23}

\begin{document}

\def\onehalf{\frac{1}{2}}
\def\etal{et.\/ al.\/}
\newcommand{\bydef}{\triangleq}
\newcommand{\tr}{{\it{tr}}}
\def\SNR{{\textsf{SNR}}}
\def\bydef{:=}
\def\bba{{\mathbb{a}}}
\def\bbb{{\mathbb{b}}}
\def\bbc{{\mathbb{c}}}
\def\bbd{{\mathbb{d}}}
\def\bbee{{\mathbb{e}}}
\def\bbff{{\mathbb{f}}}
\def\bbg{{\mathbb{g}}}
\def\bbh{{\mathbb{h}}}
\def\bbi{{\mathbb{i}}}
\def\bbj{{\mathbb{j}}}
\def\bbk{{\mathbb{k}}}
\def\bbl{{\mathbb{l}}}
\def\bbm{{\mathbb{m}}}
\def\bbn{{\mathbb{n}}}
\def\bbo{{\mathbb{o}}}
\def\bbp{{\mathbb{p}}}
\def\bbq{{\mathbb{q}}}
\def\bbr{{\mathbb{r}}}
\def\bbs{{\mathbb{s}}}
\def\bbt{{\mathbb{t}}}
\def\bbu{{\mathbb{u}}}
\def\bbv{{\mathbb{v}}}
\def\bbw{{\mathbb{w}}}
\def\bbx{{\mathbb{x}}}
\def\bby{{\mathbb{y}}}
\def\bbz{{\mathbb{z}}}
\def\bb0{{\mathbb{0}}}

\def\bydef{:=}
\def\ba{{\mathbf{a}}}
\def\bb{{\mathbf{b}}}
\def\bc{{\mathbf{c}}}
\def\bd{{\mathbf{d}}}
\def\bee{{\mathbf{e}}}
\def\bff{{\mathbf{f}}}
\def\bg{{\mathbf{g}}}
\def\bh{{\mathbf{h}}}
\def\bi{{\mathbf{i}}}
\def\bj{{\mathbf{j}}}
\def\bk{{\mathbf{k}}}
\def\bl{{\mathbf{l}}}
\def\bm{{\mathbf{m}}}
\def\bn{{\mathbf{n}}}
\def\bo{{\mathbf{o}}}
\def\bp{{\mathbf{p}}}
\def\bq{{\mathbf{q}}}
\def\br{{\mathbf{r}}}
\def\bs{{\mathbf{s}}}
\def\bt{{\mathbf{t}}}
\def\bu{{\mathbf{u}}}
\def\bv{{\mathbf{v}}}
\def\bw{{\mathbf{w}}}
\def\bx{{\mathbf{x}}}
\def\by{{\mathbf{y}}}
\def\bz{{\mathbf{z}}}
\def\b0{{\mathbf{0}}}
\def\opt{\mathsf{OPT}}
\def\on{\mathsf{ON}}
\def\off{\mathsf{OFF}}
\def\bA{{\mathbf{A}}}
\def\bB{{\mathbf{B}}}
\def\bC{{\mathbf{C}}}
\def\bD{{\mathbf{D}}}
\def\bE{{\mathbf{E}}}
\def\bF{{\mathbf{F}}}
\def\bG{{\mathbf{G}}}
\def\bH{{\mathbf{H}}}
\def\bI{{\mathbf{I}}}
\def\bJ{{\mathbf{J}}}
\def\bK{{\mathbf{K}}}
\def\bL{{\mathbf{L}}}
\def\bM{{\mathbf{M}}}
\def\bN{{\mathbf{N}}}
\def\bO{{\mathbf{O}}}
\def\bP{{\mathbf{P}}}
\def\bQ{{\mathbf{Q}}}
\def\bR{{\mathbf{R}}}
\def\bS{{\mathbf{S}}}
\def\bT{{\mathbf{T}}}
\def\bU{{\mathbf{U}}}
\def\bV{{\mathbf{V}}}
\def\bW{{\mathbf{W}}}
\def\bX{{\mathbf{X}}}
\def\bY{{\mathbf{Y}}}
\def\bZ{{\mathbf{Z}}}
\def\b1{{\mathbf{1}}}

\def\bbA{{\mathbb{A}}}
\def\bbB{{\mathbb{B}}}
\def\bbC{{\mathbb{C}}}
\def\bbD{{\mathbb{D}}}
\def\bbE{{\mathbb{E}}}
\def\bbF{{\mathbb{F}}}
\def\bbG{{\mathbb{G}}}
\def\bbH{{\mathbb{H}}}
\def\bbI{{\mathbb{I}}}
\def\bbJ{{\mathbb{J}}}
\def\bbK{{\mathbb{K}}}
\def\bbL{{\mathbb{L}}}
\def\bbM{{\mathbb{M}}}
\def\bbN{{\mathbb{N}}}
\def\bbO{{\mathbb{O}}}
\def\bbP{{\mathbb{P}}}
\def\bbQ{{\mathbb{Q}}}
\def\bbR{{\mathbb{R}}}
\def\bbS{{\mathbb{S}}}
\def\bbT{{\mathbb{T}}}
\def\bbU{{\mathbb{U}}}
\def\bbV{{\mathbb{V}}}
\def\bbW{{\mathbb{W}}}
\def\bbX{{\mathbb{X}}}
\def\bbY{{\mathbb{Y}}}
\def\bbZ{{\mathbb{Z}}}

\def\cA{\mathcal{A}}
\def\cB{\mathcal{B}}
\def\cC{\mathcal{C}}
\def\cD{\mathcal{D}}
\def\cE{\mathcal{E}}
\def\cF{\mathcal{F}}
\def\cG{\mathcal{G}}
\def\cH{\mathcal{H}}
\def\cI{\mathcal{I}}
\def\cJ{\mathcal{J}}
\def\cK{\mathcal{K}}
\def\cL{\mathcal{L}}
\def\cM{\mathcal{M}}
\def\cN{\mathcal{N}}
\def\cO{\mathcal{O}}
\def\cP{\mathcal{P}}
\def\cQ{\mathcal{Q}}
\def\cR{\mathcal{R}}
\def\cS{\mathcal{S}}
\def\cT{\mathcal{T}}
\def\cU{\mathcal{U}}
\def\cV{\mathcal{V}}
\def\cW{\mathcal{W}}
\def\cX{\mathcal{X}}
\def\cY{\mathcal{Y}}
\def\cZ{\mathcal{Z}}

\def\sfA{\mathsf{A}}
\def\sfB{\mathsf{B}}
\def\sfC{\mathsf{C}}
\def\sfD{\mathsf{D}}
\def\sfE{\mathsf{E}}
\def\sfF{\mathsf{F}}
\def\sfG{\mathsf{G}}
\def\sfH{\mathsf{H}}
\def\sfI{\mathsf{I}}
\def\sfJ{\mathsf{J}}
\def\sfK{\mathsf{K}}
\def\sfL{\mathsf{L}}
\def\sfM{\mathsf{M}}
\def\sfN{\mathsf{N}}
\def\sfO{\mathsf{O}}
\def\sfP{\mathsf{P}}
\def\sfQ{\mathsf{Q}}
\def\sfR{\mathsf{R}}
\def\sfS{\mathsf{S}}
\def\sfT{\mathsf{T}}
\def\sfU{\mathsf{U}}
\def\sfV{\mathsf{V}}
\def\sfW{\mathsf{W}}
\def\sfX{\mathsf{X}}
\def\sfY{\mathsf{Y}}
\def\sfZ{\mathsf{Z}}

\def\bydef{:=}
\def\sfa{{\mathsf{a}}}
\def\sfb{{\mathsf{b}}}
\def\sfc{{\mathsf{c}}}
\def\sfd{{\mathsf{d}}}
\def\sfee{{\mathsf{e}}}
\def\sfff{{\mathsf{f}}}
\def\sfg{{\mathsf{g}}}
\def\sfh{{\mathsf{h}}}
\def\sfi{{\mathsf{i}}}
\def\sfj{{\mathsf{j}}}
\def\sfk{{\mathsf{k}}}
\def\sfl{{\mathsf{l}}}
\def\sfm{{\mathsf{m}}}
\def\sfn{{\mathsf{n}}}
\def\sfo{{\mathsf{o}}}
\def\sfp{{\mathsf{p}}}
\def\sfq{{\mathsf{q}}}
\def\sfr{{\mathsf{r}}}
\def\sfs{{\mathsf{s}}}
\def\sft{{\mathsf{t}}}
\def\sfu{{\mathsf{u}}}
\def\sfv{{\mathsf{v}}}
\def\sfw{{\mathsf{w}}}
\def\sfx{{\mathsf{x}}}
\def\sfy{{\mathsf{y}}}
\def\sfz{{\mathsf{z}}}
\def\sf0{{\mathsf{0}}}

\def\Nt{{N_t}}
\def\Nr{{N_r}}
\def\Ne{{N_e}}
\def\Ns{{N_s}}
\def\Es{{E_s}}
\def\No{{N_o}}
\def\sinc{\mathrm{sinc}}
\def\dmin{d^2_{\mathrm{min}}}
\def\vec{\mathrm{vec}~}
\def\kron{\otimes}
\def\Pe{{P_e}}
\newcommand{\expeq}{\stackrel{.}{=}}
\newcommand{\expg}{\stackrel{.}{\ge}}
\newcommand{\expl}{\stackrel{.}{\le}}
\def\SIR{{\mathsf{SIR}}}

\def\nn{\nonumber}

\maketitle

\begin{abstract}
 We consider a single source-destination pair, where 
 information updates (in short, updates) arrive at the source at arbitrary time instants. For each update, its size, i.e. the service time required for complete transmission to the destination, is also arbitrary.
 At any time, age of information (AoI) is equal to the difference between the current time, and the arrival time of the latest update (at the source) that has been completely transmitted (to the destination). AoI quantifies the staleness of the update (information) at the destination. 
 The goal is to find a causal (i.e. online) scheduling policy that minimizes the time-average of AoI, where the possible decisions at any time are i) whether to preempt the update under transmission upon arrival of a new update, and ii) if no update is under transmission, then choose which update to transmit (among the available updates). 
 
 In this paper, we propose a causal policy called SRPT$^+$ that at each time, i) preempts the update under transmission if a new update arrives with a smaller size (compared to the remaining size of the update under transmission), and ii) if no update is under transmission, then begins to transmit the update for which the ratio of the reduction in AoI upon complete transmission (if not preempted in future) and the remaining size, is maximum. We characterize the performance of SRPT$^+$ using the metric called the competitive ratio, i.e. the ratio of the average AoI of causal policy and the average AoI of an optimal offline policy (that knows the entire input in advance), maximized over all possible inputs. We show that the competitive ratio of SRPT$^+$ is at most $4$. Further, we propose a simpler policy called SRPT$^L$, that
 i) preempts the update under transmission if a new update arrives with a smaller size (compared to the remaining size of the update under transmission), and ii) if no update is under transmission, then begins to transmit the update with the latest arrival time. We show that the competitive ratio of SRPT$^L$ is at most $29$.
\end{abstract}



\section{Introduction}
\vspace{3ex}

Online scheduling is a classical problem, where jobs with different processing requirements (sizes) arrive at arbitrary times, and an algorithm has to make its scheduling decisions using only the causal information. Typical performance metrics for which online scheduling problem has been considered are flow time \cite{schrage1968proof,scully2020simple,leonardi2007approximating,bansal2009weighted}, 
completion time \cite{hall1997scheduling,khuller2019select}, makespan \cite{graham1966bounds,goel1999stochastic}, co-flow \cite{shafiee, sincronia, twoapprox,bhimaraju2020non}, and several others. For many of these problems, either an optimal or near optimal algorithms are known. For example, the shortest remaining processing time (SRPT) algorithm is optimal for minimizing the flow time for a single server \cite{schrage1968proof}, while its multi-server counterpart has the optimal (order-wise) competitive ratio \cite{leonardi2007approximating}. 

Even though the mentioned performance metrics model a wide variety of problem settings, more modern paradigms such as internet-of-things, cyber-physical systems, etc., require a metric that can capture {\it information timeliness}. Information timeliness refers to the frequency at which fresh information is received at a controller/monitor. A compelling use case is a smart car, that is equipped with multiple sensors and a centralized monitor. Updates are generated at each sensor continuously that are needed to be made available at the monitor without being too `old'. 
An elegant and popular metric to capture information timeliness is known as the \emph{age of information (AoI)} \cite{kaul2012status,kaul2012real}, which at any time is equal to the time elapsed since the generation time of the latest update of the source that has been completely transmitted to the monitor (destination). 
The effective metric is then the time-average of the AoI (in short, average AoI). We illustrate AoI plot for a particular example (Example \ref{ex:intro}) in Figure \ref{fig:example-intro}.
\begin{example} \label{ex:intro}
	Let three updates are generated at the source, at time $g_1=0$, $g_2=2$ and $g_3=4$, respectively. Let an algorithm completely transmits the three updates by time $r_1=1$, $r_2=3$ and $r_3=5$, respectively.
	As shown in Figure \ref{fig:example-intro}, at any time $t$ in interval $[r_1,r_2)=[1,3)$, the AoI is $t-g_1=t-0=t$, while at time $t$ in interval $[r_2,r_3)=[3,5)$, the AoI is $t-g_2=t-2$. At time $t\ge r_3=5$, the AoI is $t-g_3=t-4$.
\end{example}
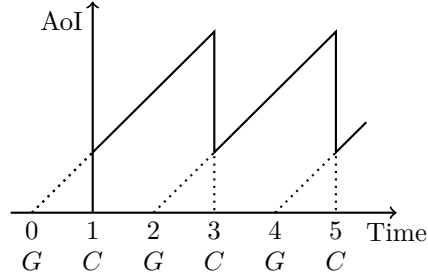
\begin{figure}
	\begin{center}
		\begin{tikzpicture}[thick,scale=0.8, every node/.style={scale=1}]
			\draw[->] (-0.35,0) to (6,0) node[below]{Time};
			\draw[->] (1,0) node[below]{$1$} to (1,3.5) node[below left]{AoI};
			
			\draw[thick] 
			(1,1) to (3,3) to (3,1) to (5,3) to (5,1) to (5.5,1.5);
			
			\draw[dotted] (0,0) node[below]{$0$} to (1,1); 
			\draw[dotted] (0,0) to (1,1);
			\draw[dotted] (2,0) node[below]{$2$} to (3,1) to (3,0) node[below]{$3$};
			\draw[dotted] (4,0) node[below]{$4$} to (5,1) to (5,0) node[below]{$5$}; 
			
			\draw (0,-0.8) node{$G$};
			\draw (1,-0.8) node{$C$};
			\draw (2,-0.8) node{$G$};
			\draw (3,-0.8) node{$C$};
			\draw (4,-0.8) node{$G$};
			\draw (5,-0.8) node{$C$};
			
%
%
%
			
		\end{tikzpicture}
		\caption{Age of information (AoI) depends on the generation time of latest completely transmitted update. Here, $G$ denotes update generation, while $C$ denotes transmission completion.
			\vspace{-4ex} 
		} 
		\label{fig:example-intro} 
	\end{center}
\end{figure}

AoI scheduling has two critical properties: i) all updates need not be processed, and ii) the AoI metric depends on the time-difference between the instants at which two consecutive updates are processed completely. Property i) makes the problem combinatorial, in the sense that an algorithm has to decide which subset of updates should be processed, while property ii) correlates the decisions across time.
Together, they make AoI scheduling fundamentally different and challenging compared to the classical scheduling for minimizing flow time, completion time, makespan, etc.

Since the definition of AoI in \cite{kaul2012real,kaul2012status}, a large body of work has been reported on scheduling for minimizing the AoI, however, the analysis is mostly restricted to stochastic inputs, as well as for some fixed scheduling algorithms. 
For example, for a single source-destination pair, \cite{kaul2012real,kaul2012status,champati2018statistical, inoue2019general,kavitha2021controlling} considered queuing models such as M/M/1 (where updates arrive with exponentially distributed inter-arrival times and sizes\footnote{Size (or equivalently, transmission/processing/service time) of an update is equal to the total time required to completely transmit/process/serve the update.}, and at any time at most one update can be transmitted) and D/G/1 (where inter-arrival times of updates are deterministic, sizes may follow any general distribution, and at any time at most one update can be transmitted), 
and analyzed the AoI for fixed scheduling policies such as the first-come-first-serve (FCFS), last-come-first serve (LCFS), etc. 
For a multi-source M/M/1 queuing models, where updates arrive from multiple sources and at any time, at most one update (from one source) can be transmitted to the destination,  \cite{multisourceyates, multisource} analyzed the distributional properties of the AoI for each source, with fixed scheduling policies such as FCFS, LCFS, etc. 

In prior work, optimality results (with respect to AoI minimization) for causal policies are known only for simpler settings, such as when there is a single source that can generate updates at any time \cite{sun2017update}, or when inter-generation time distribution is exponential. For complicated models such as with multiple sources, the performance guarantees for causal policies have been established mostly in terms of their \emph{competitive ratio} that is defined as the ratio of the cost of a causal policy, and an optimal offline policy that knows the entire update generation times and their sizes in advance.  
For example, in a discrete-time system with multiple sources, 
when the update inter-generation time and size 
distributions are geometric, 
\cite{kadota2019minimizing} showed that a causal randomized policy has a constant competitive ratio. 
Recently, similar results were shown in \cite{saurav20223,saurav2022scheduling} for continuous-time systems where update inter-generation time and size distributions belong to a class of general distributions. 

In a major departure from prior work on AoI, in this paper, we consider the AoI minimization problem for the arbitrary arrival model, where updates arrive at a source at arbitrary times and have arbitrary sizes. 
The goal is to find a causal policy that minimizes the AoI, where the possible decisions at any time are 
i) whether to preempt an ongoing update transmission on generation of a new update, and ii) if no update is under transmission, then choose which update to transmit among the available updates. 
Similar to prior work \cite{kadota2018scheduling,kadota2019minimizing,saurav20223,saurav2022scheduling,saurav2021minimizing}, we characterize the performance of any causal policy using the metric of competitive ratio. 

For a single source-destination pair with arbitrary arrival model, the SRPT policy that is optimal for minimizing the flow-time, also intuitively appears to be a reasonable policy for minimizing the AoI. 
However, as we show in Example \ref{ex:SRPT-not-OPT} (Section \ref{sec:sysmodel}), for the considered problem, the competitive ratio of SRPT is arbitrarily large. 

Another natural policy for AoI minimization 
is a greedy policy that at each time $t$, chooses to transmit that update for which the ratio of the reduction in AoI upon its completion (if not preempted) and the remaining size at time $t$, is maximum.
In the arbitrary arrival model, one can construct update generation sequences such that this greedy policy repeatedly preempts the update under transmission without transmitting any update completely, for a long time. Recall that the AoI decreases only when an update is completely transmitted. Hence, establishing competitive ratio guarantees for the greedy policy is difficult.



\paragraph*{\textbf{Contributions}} 
As described above, both the SRPT and the greedy policy on their own may not have good competitive ratio performance. 
However, for certain classes of inputs (update generation sequences), individually, both have complementarily appealing features. With this motivation,
we propose a policy called \emph{SRPT$^+$}, that at any point in time, either uses SRPT, or the greedy policy depending on the state of the system. 
In particular, when an update $i$ is under transmission, SRPT$^+$ preempts the update $i$ if (and only if) a new update $j$ is generated with size at most the remaining size of the update under transmission, and begins to transmit update $j$ (similar to SRPT). On the other hand, when no update is being transmitted, following the greedy policy, SRPT$^+$ begins to transmit the update for which the ratio of the reduction in AoI (upon its complete transmission) and the remaining size, is maximum.

The basic idea behind SRPT$^+$ is that the update that provides maximum reduction in AoI per unit transmission time (until completion) should be transmitted with priority. At the same time, preemption should not delay the completion time of updates. 
We show that for a single source destination pair with arbitrary arrival model, the competitive ratio for SRPT$^+$ is at most $4$. Essentially, we show that for each update $i$ generated at the source, the delay between its generation time and the earliest time when SRPT$^+$ completely transmits update $i$, or an update generated after it, is at most $4$ times the corresponding quantity for an optimal offline policy $\pi^\star$. 
Compared to prior work on AoI minimization, this is the first result that provides a constant competitive ratio upper bound for the arbitrary arrival model.

Subsequently, we propose a simpler policy than SRPT$^+$, called SRPT$^L$ that 
at any time when no update is being transmitted, begins to transmit the latest generated update, irrespective of its size or AoI reduction. 
Moreover, SRPT$^L$ preempts the update under transmission only if a 
new update is generated with smaller size (compared to the remaining size of the update under transmission). We show that the competitive ratio of SRPT$^L$ is at most $29$.



\section{System Model} \label{sec:sysmodel}

We consider a source-destination pair, where updates are generated at the source at arbitrary time instants $g_1,g_2,\cdots$, with arbitrary sizes (transmission times required to completely transmit the updates) $s_1,s_2,\cdots$ (where $g_i,s_i\ge 0$, $\forall i$). 
At any time, the source may transmit at most one update to the destination, however, an update under transmission can be preempted at any time, to allow transmission of a different update. 

Let $\Pi$ denote the set of all causal (online) scheduling policies (in short, policies) that at any time $t$, using only the available causal information, decides which update the source transmits at time $t$. 
For any policy $\pi\in\Pi$, the \emph{age of information (AoI)} at time $t$ is $\Delta_\pi(t)=t-\lambda^\pi(t)$, where $\lambda^\pi(t)$ denotes the generation time of the latest update of the source that has been completely transmitted until time $t$, under policy $\pi$. Thus, the AoI decreases only when an update is completely transmitted under policy $\pi$. The (time) average AoI ($\av{AoI}$) 
is defined as $\aaoi_\pi=\lim_{t\to\infty}\int_{0}^{t}\Delta_\pi(\tau)d\tau/t$.

The objective is to find a causal policy $\pi\in\Pi$ that minimizes the $\av{AoI}$. 
However, instead of directly finding the optimal causal policy, in this paper, we consider a particular causal policy called SRPT$^+\in\Pi$, and show that its competitive ratio (Definition \ref{def:CR}) is at most $4$.
\begin{definition} \label{def:CR} 
	Let $\cI=\{(g_1,s_1),(g_2,s_2),\cdots\}$ denote a particular sequence of updates that are generated at the source. Moreover, let $\pi^\star$ denote an optimal offline policy (for minimizing $\av{AoI}$) 
	that knows $\cI$ ahead of time.
	For any causal policy $\pi\in\Pi$, its competitive ratio is 
	$\textsc{CR}_{\pi}=\max_{\cI} \{\aaoi_\pi(\cI)/\aaoi_{\pi^\star}(\cI)\}$.
\end{definition}

To illustrate the challenges in minimizing the $\av{AoI}$, 
we first show that the competitive ratio of the well-known SRPT policy is arbitrarily large. In the considered setting, with SRPT, at each time $t$, the update with the least remaining time (difference of $s_i$, and the total time for which update $i$ has been under transmission until time $t$) is transmitted. 

\begin{example} \label{ex:SRPT-not-OPT}
	Let $\cI_1=\{(\epsilon/(m-i),1/2)|\forall i\in\{0,\cdots,m-1\}\}$, where $m$ is a positive integer, and $\epsilon>0$. Also, let $\cI_2=\{(i,1)|\forall i\in\bbN\}$.  
	Consider the setting where the update generation sequence is $\cI=\cI_1\cup\cI_2$, and the initial AoI is $\Delta(0)=0$. 
	
	Note that for any update $i\in\cI_1$ and $j\in\cI_2$, the size $s_i<s_j$. 
	Hence, before transmitting any update belonging to $\cI_2$, SRPT completely transmits all the updates in $\cI_1$. Moreover, the sum of the sizes of all the updates in $\cI_1$ is $m/2$. Thus, any update belonging to $\cI_2$ can begin to transmit under SRPT only at time $t\ge m/2$. 
	Also, notice that the generation time of
	the last update in $\cI_1$ is $\epsilon$. 
	Combining the above two facts, we get that under SRPT, at any time $t\le m/2$, the AoI $\Delta_{SRPT}(t)\ge t-\epsilon$. Thus, choosing $m\to\infty$ and $\epsilon\to 0^+$, we get that under SRPT, the $\av{AoI}$ 
	is $\aaoi_{SRPT}(\cI)\to\infty$.
	
	In contrast, consider a policy $\pi$, that at any time $t$ when no update is being transmitted, begins to transmit the latest generated update (if any), until completion (without preemption). 
	Clearly, for the considered update generation sequence $\cI$, $\pi$ transmits each update in $\cI_2$, starting immediately after its generation. Thus, under policy $\pi$, with input $\cI$, the $\av{AoI}$ 
	is $\aaoi_\pi(\cI)=3/2$, which implies that the $\av{AoI}$ 
	for an optimal offline policy $\pi^\star$ is $\aaoi_{\pi^\star}(\cI)\le 3/2$ (i.e. finite).
	Hence, the competitive ratio of SRPT is $\textsc{CR}_{SRPT}=\max_\cI\{\aaoi_{SRPT}(\cI)/\aaoi_{\pi^\star}(\cI)\}\to\infty$.
\end{example} 

\subsection*{Notations}

We denote the generation time of the $i^{th}$ update (equivalently, update $i$) at the source by $g_i$ (thus, the relations $i<j$ and $g_{i}< g_{j}$ are equivalent). 
With respect to policy $\pi$, define $b_i^\pi\ge g_i$ as the earliest time instant when $\pi$ begins to transmit an update $j$ with generation time $g_j\ge g_i$ ($j$ may or may not be equal to $i$ because all updates need not be transmitted in order). Further, define $r_i^\pi\ge b_i^\pi$ as the earliest time instant when the transmission of an update $k$ with generation time $g_k\ge g_i$ completes under policy $\pi$ (update $k$ may or may not be same as update $i$, or the update $j$ whose transmission began at $b_i^\pi$, due to preemption). 
Note that 
the generation time of the latest completely transmitted update i.e. $\lambda^\pi(t)< g_i$, $\forall t<r_i^\pi$. 
As shown in Figure \ref{fig:age-plot-notations},  if policy $\pi$ completely transmits update $i+1$ without or before completely transmitting update $i$, then $r_i^\pi=r_{i+1}^\pi$. 

Let $\delta_i=g_i-g_{i-1}$ denote the inter-generation time of successive updates $i-1$ and $i$, which is independent of policy $\pi$. Define $\nu_i^\pi=r_i^\pi-g_i 
=(r_i^\pi-b_i^\pi)+(b_i^\pi-g_i)=d_i^\pi+w_i^\pi$, where $d_i^\pi=r_i^\pi-b_i^\pi$ and $w_i^\pi=b_i^\pi-g_i$. Note that $\nu_i^\pi$, $w_i^\pi$ and $d_i^\pi$ are update-based metrics that are defined for each generated update. Figure \ref{fig:age-plot-notations} shows a sample AoI plot for policy $\pi\in\Pi$, along with the quantities defined above.

Next, we derive a general expression for the $\av{AoI}$ of any causal or offline policy $\pi$, in terms of the quantities defined so far.

\section{General Expression for Average AoI}

\begin{lemma} \label{lemma:aaoi}
	For any policy $\pi$ (causal or offline), the $\av{AoI}$ is 
		\begin{align} \label{eq:aaoi-general}
			\aaoi_\pi&=\lim_{t\to\infty}\frac{1}{t}\sum_{i=1}^{R(t)}\left[\frac{\delta_{i}^2}{2}+\delta_i\nu_i^\pi\right]=\lim_{t\to\infty}\frac{1}{t}\left[\sum_{i=1}^{R(t)}\frac{\delta_{i}^2}{2}+\sum_{i=1}^{R(t)}\delta_i w_i^\pi+\sum_{i=1}^{R(t)}\delta_i d_i^\pi\right],
	\end{align} 
where $R(t)$ denotes the number of updates generated until time $t$.
\end{lemma} 
\begin{proof}
	See Appendix \ref{appendix:proof-lemma-aaoi}.
\end{proof}
\begin{figure}
	\begin{center}
		\begin{tikzpicture}[thick,scale=0.8, every node/.style={scale=1}]
		\draw[->] (-0.25,0) to (6.8,0) node[below]{time ($t$)};
		\draw (0.2,0) node[below]{0};
		\draw[->] (0.35,-0.25) to (0.35,3.6) node[below left]{$\Delta_{\pi}(t)$};
		\draw (0.35,0) to (2.4,2.05) to (2.4,0.9) to (5,3.5) to (5,0.8) to (5.75,1.55) to (5.8,1.6) ;
		
		\fill[pattern=north west lines, pattern color=red] (0.35,0) to (2.4,2.05) to (2.4,0.9) to (1.5,0) to (0.35,0);
		\fill[pattern=north east lines, pattern color=green!50!gray] (1.5,0) to (5,3.5) to (5,1.8) to (3.2,0) to (1.5,0);
		 
		\fill[pattern=north west lines, pattern color=blue] (5,1.8) to (3.2,0) to (4.2,0) to (5,0.8) to (5,1.8); 
		
		
		\draw[loosely dotted] (6,1) to (6.7,1); 

		\draw (1.5,-0.1) node[below]{$g_{1}$} to (1.5,0.1);
		\draw (3.2,-0.1) node[below]{$g_{2}$} to (3.2,0.1);
		\draw (4.2,-0.1) node[below]{$g_{3}$} to (4.2,0.1); 
		
		\draw (2.4,-0.1) node[below]{$r_{1}^\pi$} to (2.4,0.1);
		\draw (5,-0.1) to (5,0.1);
		\draw (4.7,-0.1) node[below right]{$r_{2}^\pi,r_{3}^\pi$};
		
		\draw[dashed] (1.5,0.1) to (1.5,2.2);
		\draw[dashed] (3.2,0.1) to (3.2,3.6);
		\draw[dashed] (4.2,0.1) to (4.2,2.7);
		
		\draw[dashed] (1.5,0) to (2.4,0.9) to (2.4,0.1);
		\draw[dashed] (3.2,0) to (5,1.8); 
		\draw[dashed] (4.2,0) to (5,0.8) to (5.,0.1);
		
        \draw[|<->] (0.35,-0.9) -- (1.5,-0.9) node[rectangle,inner sep=-1pt,midway,fill=white]{$\delta_{1}$}; 
        \draw[|<->] (1.5,-0.9) -- (3.2,-0.9) node[rectangle,inner sep=-1pt,midway,fill=white]{$\delta_{ 2}$};
        \draw[|<->|] (3.2,-0.9) -- (4.2,-0.9) node[rectangle,inner sep=-1pt,midway,fill=white]{$\delta_{3}$};
		
        \draw[|<->|] (1.5,2.2) -- (2.4,2.2) node[rectangle,inner sep=-1pt,midway,fill=white]{$\nu_{1}^\pi$}; 
        \draw[|<->|] (3.2,3.6) -- (5,3.6) node[rectangle,inner sep=-1pt,midway,fill=white]{$\nu_{2}^\pi$};
        \draw[|<->|] (4.2,2) -- (5,2) node[rectangle,inner sep=-1pt,midway,fill=white]{$\nu_{3}^\pi$};
		
		
		\end{tikzpicture}
		\caption{Sample AoI plot under policy $\pi$. 
		\vspace{-4ex} 
		} 
		\label{fig:age-plot-notations} 
	\end{center}
\end{figure}
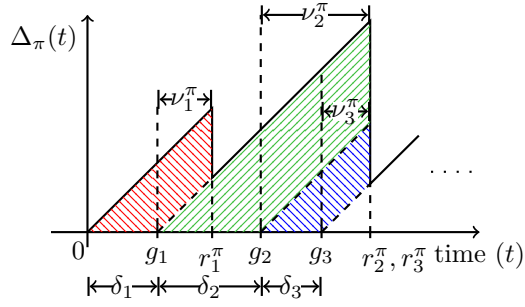

Notably, Lemma \ref{lemma:aaoi} expresses the $\av{AoI}$ for any policy $\pi$ (causal or offline) in terms of $\delta_i$'s that are independent of policy $\pi$, and update-based metrics $\nu_i^\pi$'s (i.e. $w_i^\pi$'s and $d_i^\pi$'s). Also, 
from \eqref{eq:aaoi-general}, we get that the $\av{AoI}$ for an optimal offline policy $\pi^\star$ (that knows the update generation sequence in advance) 
is 
	\begin{align} \label{eq:aaoi-opt}
		\aaoi_{\pi^\star}=\lim_{t\to\infty}\frac{1}{t}\sum_{i=1}^{R(t)}\left[\frac{\delta_{i}^2}{2}+\delta_i\nu_i^\star\right],
	\end{align}
   where $\nu^\star=r_i^\star-g_i$, and $r_i^\star$ is the earliest time instant when $\pi^\star$ completes the transmission of an update $j\ge i$. Lemma \ref{lemma:lb-nu-star} shows that $\nu_i^\star\ge\min_{j\ge i}\{g_j+s_j\}-g_i$, for each update $i$.
   
   \begin{lemma} \label{lemma:lb-nu-star}
   	Under any policy $\pi$ (causal or offline), for each update $i$, $\nu_i^\pi=r_i^\pi-g_i\ge \nu_i^{\min}$, where $\nu_i^{\min}=\min_{j\ge i}\{g_j+s_j\}-g_i\ge 0$.
   \end{lemma}
   \begin{proof}
   	Note that the earliest time instant by which an update $i$ generated at time $g_i$ can be completely transmitted is  $g_i+s_i$, where $s_i=s_i(g_i)$ is the size of update $i$ at generation. Therefore, the earliest time instant when an update $j\ge i$ can be completely transmitted is $\min_{j\ge i}\{g_j+s_j\}$. Hence, for any causal or offline policy $\pi$, 
   	we have $r_i^\pi\ge \min_{j\ge i}\{g_j+s_j\}$, which implies that $\nu_i^\pi=r_i^\pi-g_i\ge \min_{j\ge i}\{g_j+s_j\}-g_i=\nu_i^{\min}$. 
   \end{proof}

In the next section, we propose a novel policy called SRPT$^+$ that has features of both a greedy policy and SRPT, and show that the $\av{AoI}$ \eqref{eq:aaoi-general} for SRPT$^+$ is at most $4$ times the $\av{AoI}$ \eqref{eq:aaoi-opt} for an optimal offline policy $\pi^\star$.

\section{SRPT$^+$ Policy} 

At time $t$, let $s_i(t)$ denote the 
remaining size of update $i$, and $\lambda(t)$ be the generation time of the latest update completely transmitted until time $t$. If the remaining size $s_i(t)$ of update $i$ is transmitted from time $t$ onwards (and not preempted ever), then its transmission completes at time $t+s_i(t)$, causing an instantaneous reduction in AoI of $\max\{g_i-\lambda(t),0\}$ units at time instant $t+s_i(t)$. 
Thus, for each update $i$ available at time $t$, i.e., $s_i(t)>0$, we define an index
\begin{align} \label{eq:slope}
	\gamma_i(t)=\frac{\max\{g_i-\lambda(t),0\}}{s_i(t)},
\end{align}
which represents the ratio of the benefit obtained by transmitting update $i$, and the cost (time required to transmit it).
Let $\Gamma(t)=\{i|\gamma_i(t)>0,s_i(t)>0\}$ denote the set of all updates available at time $t$ such that $\gamma_i(t)>0$. 

Using the index $\gamma_i(t)$, we propose a policy called \emph{SRPT$^+$} (Algorithm \ref{algo:AoI-SRPT}) that at any time $t$, makes one of the following two decisions.  i) If no update is under transmission, then it begins to transmit the update $i^\star(t)=\arg\max_{j\in\Gamma(t)}\gamma_j(t)$ (and idles if the set $\Gamma(t)$ is empty), and ii) if an update $i$ is under transmission, then update $i$ is preempted by update $j$ if and only if a new update $j$ is generated with size $s_j=s_j(t)\le s_i(t)$. 

\begin{algorithm}
	\caption{SRPT$^+$ Policy.} 
	\label{algo:AoI-SRPT}
	\begin{algorithmic}[1]
		\STATE At any time $t$, 
		\IF{an update $i$ is under transmission}
		\IF{a new update $j$ is generated with size $s_j=s_j(t)\le s_i(t)$}
		\STATE preempt update $i$ and begin to transmit update $j$.
		\ELSE
		\STATE continue transmitting update $i$.
		\ENDIF
		\ELSIF{the set $\Gamma(t)$ is non-empty}
		\STATE begin to transmit update $j=\arg\max_{i\in\Gamma(t)}\gamma_i(t)$ (ties broken arbitrarily).
		\ELSE
		\STATE idle until a new update is generated.
		\ENDIF
	\end{algorithmic}
\end{algorithm}

\begin{remark} \label{remark:greedy+}
	Note that for any update $i$ and $j$ with $g_i<g_j$, if $s_i(t)\ge s_j(t)$, then 
	\begin{align*}
		\gamma_i(t)=\frac{\max\{g_i-\lambda(t),0\}}{s_i(t)}<\frac{\max\{g_j-\lambda(t),0\}}{s_j(t)}=\gamma_j(t).
	\end{align*}
	Thus, under SRPT$^+$, when an update $j$ preempts an update $i$ at time $t$, this implicitly means that 
	$\gamma_j(t)>\gamma_i(t)$ and $j\in\Gamma(t)$. This shows that SRPT$^+$ is essentially a greedy policy with respect to index \eqref{eq:slope}, except that to preempt an update $i$ at time $t$ (with remaining size $s_i(t)$), it additionally requires that the size $s_j(t)$ of update $j$ (that is to preempt update $i$) is no more than $s_i(t)$. Hence, at the time instant $t$ when the transmission of an update $i$ begins, its index $\gamma_i(t)$ is maximum among all the updates. Also, preemption does not increase the earliest time after time $t$, when an update from $\Gamma(t)$ gets completely transmitted. 
\end{remark}

%

\begin{remark}
	SRPT$^+$ transmits an update $i$ an time $t$, only if it belongs to $\Gamma(t)$, i.e. $\gamma_i(t)>0$. In fact, at time $t$, if update $i$ does not belongs to $\Gamma(t)$ ($\gamma_i(t)=0$), 
	then SRPT$^+$ never transmits it in the future as well. This is because $\lambda(t)$ is a non-decreasing function of $t$, and the generation time $g_i$ is fixed for update $i$. Thus, at any time $t$, if $\gamma_i(t)=0$ (i.e., $\max\{g_i-\lambda(t),0\}=0$), then $\gamma_i(t')=0$, $\forall t'\ge t$.
\end{remark}


We next illustrate SRPT$^+$ (Algorithm \ref{algo:AoI-SRPT}) using a pertinent example that brings out its critical properties. 
\begin{example} 
	Let the AoI at time $t=0$ be $\Delta(0)=0$, which implies $\lambda(0)=0$. 
	Consider an update generation sequence $\cI=\{(g_i,s_i)\}_i$, where the set of updates generated in time interval $[0,2]$ are $\{(0,1.45),(0.25,1.25),(0.75,1),(1,0.5),(1.25,0.3), (1.8,0.1)\}$. Let the updates be numbered (ordered) in increasing order of their generation time.
	
	Clearly, at time $t=0$, the only update available is update $1$, i.e. $(0,1)$, and no update is under transmission. 
	Moreover, the index \eqref{eq:slope} for the update available at time $t=0$ is $\gamma_1(0)=\max\{g_1-\lambda(0),0\}/s_1(0)=0$, which implies that the update does not lie in $\Gamma(0)$. Hence, SRPT$^+$ discards update $1$, and no transmission happens until time $t=0.25$, when update $2$ is generated. 
	
	Until time $t=0.25$, since no update is transmitted by SRPT$^+$, at time $t=0.25$,  $\lambda(0.25)=0$ (and AoI $\Delta(t=0.25)=t-\lambda(t)=0.25$). Thus, $\gamma_2(0.25)=\max\{g_2-\lambda(0.25),0\}/s_2(0.25)=(0.25-0)/1.25=0.2>0$. Thus, SRPT$^+$ begins to transmit update $2$ at time $t=0.25$, which remains under transmission at time $t=0.75$ with 
	remaining size $s_2(0.75)=1.25-(0.75-0.25)=0.75$, and index $\gamma_2(0.75)\approx 0.33$, when update $3$ is generated. Until time $t=0.75$, since no update has been completely transmitted by SRPT$^+$, $\lambda(0.75)=0$. For update $3$ the size at time $t=0.75$ is $s_3(0.75)=s_3=1$, and index $\gamma_3(0.75)=\max\{g_3-\lambda(0.75),0\}/s_3(0.75)=(0.75-0)/1=0.75$. Clearly, $\gamma_2(0.75)<\gamma_3(0.75)$, but since update $2$ is already under transmission and $s_2(0.75)=0.75<s_3(0.75)=1$, SRPT$^+$ continues to transmit update $2$. 
	
	At time $t=1$, update $2$ is still under transmission with remaining size $s_2(1)=1.25-(1-0.25)=0.5$, and index $\gamma_2(1)=\max\{g_2-\lambda(1),0\}/s_2(1)=(0.25-0)/0.5=0.5$. However, update $4$ is generated at time $t=1$, with size $s_4=0.5\le s_2(1)$. Hence, SRPT$^+$ preempts update $2$, and begins to transmit update $4$ (at time $t=1$). 
	
	Subsequently, at time $t=1.25$, update $4$ is under transmission with remaining size $s_4(1.25)=(0.5-(1.25-1))=0.25$ which is less than the size of update $5$ generated at time $t=1.25$. Hence, SRPT$^+$ continues to transmit update $4$ until time $t=1.5$, when its transmission completes. Thus, at time $t=1.5$, the latest completely transmitted update is update $4$, which implies $\lambda(1.5)=g_4=1$.
	
	Since, $\lambda(1.5)=g_4=1$, at time $t=1.5$, the index $\gamma_i(1.5)=\max\{g_i-1,0\}/s_i(1.5)=0$, for each update $i$ generated until time $t=1$. Thus, the only update with a positive index at time $t=1.5$ (i.e., the only update in $\Gamma(1.5)$) is update $5$ (i.e. $(1.25,0.3)$) with index $\gamma_5(1.5)=\max\{g_5-\lambda(1.5),0\}/s_5(1.5)=(1.25-1)/0.3\approx 0.83$. Hence, SRPT$^+$ begins to transmit update $5$ starting at time $t=1.5$, which completes at time $t=1.8$. Thus, at time $t=1.8$, $\lambda(1.8)=g_5=1.25$, and AoI is $t-\lambda(t)=1.8-1.25=0.55$. 
	
	Subsequently, at time $t=1.8$, update $6$ is generated, and has index $\gamma_6(1.8)=0.55>0$. Since no other update $i$ has a positive index $\gamma_i(1.8)$ (at time $t=1.8$), SRPT$^+$ begins to transmit update $6$ at time $t=1.8$, and completes it at time $t=1.9$. Thus, at time $t=1.9$, $\lambda(1.9)=g_6=1.8$, and AoI is $t-\lambda(t)=1.9-1.8=0.1$. Following time $t=1.9$, no update is available with positive index \eqref{eq:slope}. Hence, SRPT$^+$ idles thereafter.
	
	Note that for the given update generation sequence in interval $[0,2]$, under SRPT$^+$, only three updates are completely transmitted: i) update $4$ i.e. $(1,0.5)$, completely transmitted at time $t=1.5$, ii) update $5$ i.e. $(1.25,0.3)$, completely transmitted at time $t=1.8$, and iii) update $6$ i.e. $(1.8,0.1)$, completely transmitted at time $t=1.9$. Thus, under SRPT$^+$, AoI increases linearly from $0$ at time $t=0$ to $1.5$ at time $t=1.5$. Then, as soon as the transmission of update $4$ completes, AoI decreases to $1.5-\lambda(1.5)=1.5-g_4=1.5-1=0.5$. Then, AoI increases linearly to $0.8$ at time $t=1.8$, when AoI decreases to $1.8-\lambda(1.8)=1.8-g_5=0.55$. Subsequently, AoI increases linearly to $0.65$ at time $t=1.9$, then decreases to $0.1$ when update $6$ is completely transmitted, and increases linearly thereafter. On simple computation, we find that $\av{AoI}$ for SRPT$^+$ in interval $[0,2]$ is $1.395$ time units. 
	
	Using brute force technique, we find that for the given update generation sequence in interval $[0,2]$, with an optimal offline policy $\pi^\star$, only updates $5$ and $6$ are completely transmitted. In particular, $\pi^\star$ completely transmits update $5$ over the time interval $(1.25,1.55]$, and update $6$ over the time interval $(1.8,1.9]$, incurring $\av{AoI}$ equal to $1.3825$ time units. 
\end{example}





\textbf{The main result of this section is as follows.}
\begin{theorem} \label{thm:main-result}
	The competitive ratio of SRPT$^+$ (Algorithm \ref{algo:AoI-SRPT}) is $\textsc{CR}_{\text{SRPT}^+}\le 4$.
\end{theorem}

The main technical Lemma to prove Theorem \ref{thm:main-result} are as follows.
\begin{lemma} \label{lemma:wait-delay-ub}
	Let superscripts $+$ and $\star$ respectively denote SRPT$^+$, and 
	an optimal offline policy $\pi^\star$. 
	\begin{enumerate}
		\item For any update $i$, $w_i^+\le 2 \nu_i^\star$, where $w_i^+=b_i^+-g_i$, while $\nu_i^\star=r_i^\star-g_i$. 
		\item $\sum_{i=1}^{R(t)}\delta_i d_i^+\le 2\sum_{i=1}^{R(t)}\delta_i \nu_i^\star$, where $\delta_i=g_i-g_{i-1}$, 	while $d_i^+=r_i^+-b_i^+$.
	\end{enumerate}
\end{lemma}
\begin{proof}
	See Appendix \ref{app:proof-lemma-wait-delay-ub}.
\end{proof}

\begin{proof}[Proof of Theorem \ref{thm:main-result}]
	From \eqref{eq:aaoi-general} and Lemma \ref{lemma:wait-delay-ub}, we get that the $\av{AoI}$ for SRPT$^+$ is
	\begin{align} \label{eq:for-remark}
		\aaoi_{\text{SRPT}^+} &\le \lim_{t\to\infty}\frac{1}{t}\left[\sum_{i=1}^{R(t)}\frac{\delta_{i}^2}{2}+\sum_{i=1}^{R(t)}\delta_i (2\nu_i^\star)+2\sum_{i=1}^{R(t)}\delta_i \nu_i^\star\right], \nonumber \\
		&= \lim_{t\to\infty}\frac{1}{t}\left[\sum_{i=1}^{R(t)}\frac{\delta_{i}^2}{2}+4\sum_{i=1}^{R(t)}\delta_i \nu_i^\star\right], \nonumber \\
		&\stackrel{(a)}{\le} \lim_{t\to\infty}\frac{4}{t}\left[\sum_{i=1}^{R(t)}\frac{\delta_{i}^2}{2}+\sum_{i=1}^{R(t)}\delta_i \nu_i^\star\right], \\
		&\stackrel{(b)}{=}4\cdot\aaoi_{\pi^\star}, \nonumber 
	\end{align}
	where in $(a)$ since $\delta_i^2\ge 0$, we upper bound it by $4\delta_i^2$, 
	and $(b)$ follows from \eqref{eq:aaoi-opt}. 
	\begin{remark} \label{remark:reason}
		Note that $\delta_i$'s are inter-generation time of updates that only depends on when the updates are generated, while $\nu_i^\star$ additionally depends on the sizes of updates (Lemma \ref{lemma:lb-nu-star}). Since the generation time of updates and their sizes can be arbitrary (independent of each other), in $(a)$ of \eqref{eq:for-remark}, $\delta_i\nu_i^\star$'s can be arbitrary large compared to $\delta_i^2$'s. \qedhere 
	\end{remark}
\end{proof}

Next, we analyze a simpler policy than SRPT$^+$, called SRPT$^L$, and show that its competitive ratio is at most $29$.

\section{SRPT$^L$ Policy}

Following any time instant $t$ (when the generation time of the latest completely transmitted update is $\lambda(t)$), if an update $i$ that is generated at  time $g_i>\lambda(t)$ is completely transmitted, the reduction in AoI is $g_i-\lambda(t)$. Clearly if $g_i$ is large (i.e. update $i$ is the latest generated update), the reduction in AoI is also large. Motivated by this fact, we consider a causal policy called \emph{SRPT$^L$} (Algorithm \ref{algo:AoI-LCFS}), defined as follows: 
At any time $t$, if no update is under transmission, SRPT$^L$ begins to transmit the latest  generated update $i^\star=\arg\max_i g_i$ (if $i^\star$ is yet to be transmitted completely, else wait until next update is generated). Else, if an update $i$ is under transmission at time $t$, then SRPT$^L$ preempts update $i$ and begins to transmit update $j$ (at time $t$), 
only if update $j$ is generated at time $t$ with size $s_j\le s_i(t)$.


\begin{algorithm}
	\caption{SRPT$^L$ Policy.} 
\label{algo:AoI-LCFS}
\begin{algorithmic} [1]
	\STATE At any time $t$, 
	\IF{an update $i$ is under transmission}
	\IF{a new update $j$ is generated with size $s_j(t)\le s_i(t)$}
	\STATE preempt update $i$ and begin to transmit update $j$.
	\ELSE
	\STATE continue transmitting update $i$.
	\ENDIF
	\ELSIF{latest generated update is yet to be transmitted completely} 
	\STATE begin to transmit the latest generated update. 
	\ELSE
	\STATE idle until a new update is generated.
	\ENDIF
\end{algorithmic}
\end{algorithm} 

Regarding SRPT$^L$, the fact that only a smaller update $j$ (compared to the remaining size of the update $i$ under transmission) can preempt update $i$, ensures that when SRPT$^L$ begins to transmit an update $i$, the earliest time when an update $j\ge i$ is completely transmitted (i.e. $r_i^L$), is never delayed due to preemption. 
In Lemma \ref{lemma:LCFS+properties}, we show some critical properties of SRPT$^L$.

\begin{lemma} \label{lemma:LCFS+properties}
	Let superscripts $L$ and $\star$ respectively denote SRPT$^L$ and $\pi^\star$. 
	\begin{enumerate}
		\item For any update $i$, $w_i^L\le \nu_i^{\min}\le \nu_i^\star$, where $w_i^L=b_i^L-g_i$, while $\nu_i^{\min}$ and $\nu_i^\star$ are as defined in Lemma \ref{lemma:lb-nu-star}.
		\item At the time instant $t$ when SRPT$^L$ begins to transmit an update $i$, its generation time $g_i$ is the latest among all updates generated until time $t$. That is, SRPT$^L$ either never begins to transmit an update $i$, or begins to transmit it at some time $t\in[g_i,g_{i+1})$.  
		\item If SRPT$^L$ begins to transmit an update $i$ (at time $b_i^L$), then $r_i^L=\min\{b_i^L+s_i,\min_{j\ge i+1}\{g_j+s_j\}\}$. Additionally, for such an update $i$,  $d_i^L=r_i^L-b_i^L\le \nu_i^{\min}\le \nu_i^\star$.
	\end{enumerate} 
\end{lemma}
\begin{proof}
	See Appendix \ref{app:proof-lemma-LCFS+properties}.
\end{proof}

\begin{remark}
	At any time $t$, while choosing the latest generated update $i^\star$ for transmission, SRPT$^L$ neither considers the update size $s_{i^\star}(t)$, nor the AoI at time $t$. 
\end{remark}

In Appendix \ref{app:example-LCFS}, we illustate the properties of SRPT$^L$ using the update generation sequence $\cI$ considered in Example \ref{ex:SRPT-not-OPT}.
Next, we show the main result of this section. 
\begin{theorem} \label{thm:LCFS}
	The competitive ratio of SRPT$^L$ (Algorithm \ref{algo:AoI-LCFS}) is $\textsc{CR}_{\text{SRPT}^L}\le 29$.
\end{theorem}
\begin{proof}
	See Appendix \ref{app:proof-thm-LCFS}.
\end{proof}

\section{Conclusions}
In this paper, we have considered a  scheduling problem for a single source-destination pair. The goal is to find a causal scheduling policy that minimizes the average age of information (AoI). Unlike most prior works that assume update generation sequence to be either deterministic or stochastic, we have considered an arbitrary arrival model, where updates of arbitrary sizes arrive at the source at arbitrary time instants, and at any time, the possible decisions are whether to preempt the update under transmission on arrival of a new update, and if no update is under transmission, then which update (among the available updates) to transmit. We have proposed two causal scheduling policies called SRPT$^+$ and SRPT$^L$, and upper bounded their competitive ratio by comparing against an optimal offline policy. In particular, we have shown that SRPT$^+$ and SRPT$^L$ have constant competitive ratios, at most $4$ and $29$, respectively. The remaining open question is: whether similar (constant) competitive ratio guarantees can be shown for causal policies in systems with multiple sources.

\bibliography{reflist,refs} 

\appendix
\section{Proof of Lemma \ref{lemma:aaoi}} \label{appendix:proof-lemma-aaoi} 
\begin{proof}
By definition, $\aaoi_\pi=\lim_{t\to\infty}\int_{0}^t\Delta_\pi(\tau)d\tau/t$. The main idea for showing \eqref{eq:aaoi-general} 
is to 
express the 
integral $\lim_{t\to\infty}\int_{0}^{t}\Delta_\pi(\tau)d\tau$ 
for a policy $\pi$ (causal or offline) in terms of $\delta_i$'s and $\nu_i^\pi$'s. 
As shown in Figure \ref{fig:age-plot-notations}, for a fixed policy $\pi$, the area under its AoI plot (which is equal to the integral $\lim_{t\to\infty}\int_{0}^{t}\Delta_\pi(\tau)d\tau$) can be partitioned into segments of area $\delta_i^2/2+\delta_i \nu_i^\pi$, corresponding to each update $i$. 
Thus, summing the area of the segments across all updates, we get $\lim_{t\to\infty}\int_{0}^{t}\Delta_\pi(\tau)d\tau=\sum_{i=1}^{R(t)}(\delta_i^2/2+\delta_i \nu_i^\pi)$, where $R(t)$ denotes the number of updates generated at the source until time $t$. 
Thus,
\begin{align*} 
	\aaoi_\pi=\lim_{t\to\infty}\frac{1}{t}\int_{0}^t\Delta_\pi(\tau)d\tau&=\lim_{t\to\infty}\frac{1}{t}\sum_{i=1}^{R(t)}\left[\frac{\delta_{i}^2}{2}+\delta_i\nu_i^\pi\right],  \\
	&=\lim_{t\to\infty}\frac{1}{t}\sum_{i=1}^{R(t)}\left[\frac{\delta_{i}^2}{2}+\delta_i(w_i^\pi+d_i^\pi)\right], \nonumber \\ 
	&=\lim_{t\to\infty}\frac{1}{t}\left[\sum_{i=1}^{R(t)}\frac{\delta_{i}^2}{2}+\sum_{i=1}^{R(t)}\delta_i w_i^\pi+\sum_{i=1}^{R(t)}\delta_i d_i^\pi\right]. \qedhere 
\end{align*}
\end{proof}

\section{Proof of Lemma \ref{lemma:wait-delay-ub}}
\label{app:proof-lemma-wait-delay-ub}
Before discussing the detailed proof of Lemma \ref{lemma:wait-delay-ub} in Subsections \ref{app:proof-lemma-SRPT+-1} and \ref{app:proof-lemma-SRPT+-2}, we provide its proof sketch.
\begin{proof}[Proof Sketch]
	(1) $\nu_i^\star=r_i^\star-g_i$ implies that $r_i^\star=g_i+\nu_i^\star$. Thus, there exists an update $j^\star\ge i$ that the optimal offline policy $\pi^\star$ completes transmitting at time $r_i^\star=g_i+\nu_i^\star$. 
	To prove the first result, we consider the following cases.
	i) \emph{No update is under transmission at time $g_i$.} Then the transmission of update $i$ starts immediately at generation, which implies $w_i^+=0\le \nu_i^\star$.
	ii) \emph{An update $f$ is under transmission at time $g_i$.} We consider the following two sub-cases.
	
	a) \emph{Under SRPT$^+$, update $f$ is preempted by some update $j\ge i$.} We show that if update $f$ is preempted by some update $j\ge i$, then it must happen until time $g_{j^\star}$ (generation time of update $j^\star$ that $\pi^\star$ completely transmits until time $r_{j^\star}^\star=g_i+\nu_i^\star=r_i^\star$). 
	Since $g_{j^\star}\le r_{j^\star}^\star=g_i+\nu_i^\star$ (by definition), we get that $b_i^+\le g_i+\nu_i^\star$. Thus $w_i^+=b_i^+-g_i\le\nu_i^\star$.
	
	b) \emph{Under SRPT$^+$, update $f$ is not preempted, i.e., its transmission completes (at time $r_f^+$).}  Since update $f$ is not preempted, its size $s_f(g_j)=s_f(g_i)-(g_j-g_i)<s_j$, $\forall j\ge i$. Thus, $s_f(g_i)<\min_{j\ge i}\{g_j+s_j\}-g_i=\nu_i^{\min}\le\nu_i^\star$ (Lemma \ref{lemma:lb-nu-star}), which implies that the completion time of transmission of update $f$ (under SRPT$^+$) is $r_f^+\le g_i+\nu_i^\star$. 
	Following this, to show that $w_i^+\le 2\nu_i^\star$ (i.e. $b_i^+\le g_i+2\nu_i^\star$), we show that in interval $(r_f^+,g_i+2\nu_i^\star]$ (which is of length greater than $\nu_i^\star$), SRPT$^+$ cannot be transmitting updates generated prior to update $i$ 
	all the time (and must begin to transmit an update $j\ge i$ at time $b_i^+\le 2\nu_i^\star$). 
	In particular, we show that 
	the sum of the size of updates $k<i$ 
	with index $\gamma_k(t)\ge \gamma_{j^\star}(t)$ (at any time $t\in(r_f^+,g_i+2\nu_i^\star]$) cannot be greater than $\nu_i^\star$. This is because the size of $j^\star$ is at most $\nu_i^\star$, and the reduction in AoI on directly transmitting $j^\star\ge i$ (after $r_f^+$) is greater than that possible if all updates $k<i$ 
	are completely transmitted.
	
	
	(2) For the second result, we partition the set of updates generated at the source into subsets $\cA_1, \cA_2, \cdots$, such that $b_i^+=b_{i'}^+$, (i.e., under SRPT$^+$, the earliest time instant when the transmission of an update $j\ge i$ begins is equal to the earliest time instant when the transmission of an update $j'\ge i'$ begins), if and only if $i,j\in\cA_n$ for some $n$. Then, we show that for each $n$, $\sum_{i\in\cA_n} \delta_i d_i^+\le 2\sum_{i\in\cA_n}\delta_i\nu_i^\star$. The intuition for this result is captured in Figure \ref{fig:rect}.
	%
\end{proof}
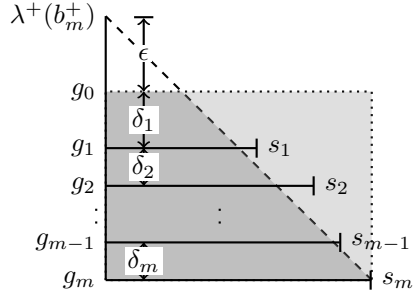
\begin{figure}
	\begin{center}
		\begin{tikzpicture}[thick,scale=0.5, every node/.style={scale=1}]
        \draw[dotted] (0,0) rectangle (7,5);
        \draw[dashed] (0,7) node[left]{$\lambda^+(b_m^+)$} to (7,0);
        
        \fill[gray,opacity=0.5] (0,0) to (0,5) to (2,5) to (7,0) to (0,0);
        \fill[gray,opacity=0.25] (7,0) to (7,5) to (2,5) to (7,0);
        
        \draw (0,0) node[left]{$g_m$} to (0,7);
        \draw[-|] (0,0) to (7,0) node[right]{$s_m$};
        \draw[-|] (0,1) node[left]{$g_{m-1}$} to (6.2,1) node[right]{$s_{m-1}$};
        \draw[-|] (0,2.5) node[left]{$g_2$} to (5.5,2.5) node[right]{$s_2$};
        \draw[-|] (0,3.5) node[left]{$g_1$} to (4,3.5) node[right]{$s_1$};
        \draw (0,5) node[left]{$g_0$}; 
        
        \draw[<->] (1,0) -- (1,1) node[rectangle,inner sep=1pt,midway,fill=white]{$\delta_{m}$};
        \draw[<->] (1,2.5) -- (1,3.5) node[rectangle,inner sep=1pt,midway,fill=white]{$\delta_{2}$};
        \draw[<->] (1,3.5) -- (1,5) node[rectangle,inner sep=1pt,midway,fill=white]{$\delta_{1}$};
        \draw[<->|] (1,5) -- (1,7) node[rectangle,inner sep=1pt,midway,fill=white]{$\epsilon$};
		
		\draw[loosely dotted] (-0.25,1.5) to (-0.25,2.1);
		\draw[loosely dotted] (3,1.5) to (3,2.1);
		
		\end{tikzpicture}
		\caption{Slope of the dashed line represents the index \eqref{eq:slope} for update $m$ (latest generated update in $\cA_n$) at time $b_m^+$, which is largest among all updates in $\cA_n$. Note that $\sum_{i\in\cA_n}\delta_i d_i^+$ is at most equal to the area of the rectangular region (total shaded region), which is at most twice the area of the darker shaded region, a lower bound on $\sum_{i\in\cA_n}\delta_i \nu_i^\star$. 
		\vspace{-4ex} 
		} 
		\label{fig:rect} 
	\end{center}
\end{figure}

\subsection{For each update $i$, $w_i^+\le 2 \nu_i^\star$, where $w_i^+=b_i^+-g_i$, while $\nu_i^\star=r_i^\star-g_i$.} \label{app:proof-lemma-SRPT+-1}
\begin{proof}
	If no update is under transmission at time $g_i$ (the generation time of update $i$), then SRPT$^+$ begins to transmit update $i$ immediately after generation. Thus, $b_i^+=g_i$, which implies $w_i^+=b_i^+-g_i=0\le 2\nu_i^\star$ (since $\nu_i^\star\ge 0$). 
	Therefore, to show that $w_i^+\le 2 \nu_i^\star$, it is sufficient to consider the case where SRPT$^+$ is transmitting some update $f$ at time $g_i$, 
	when update $i$ gets generated. 

Recall that SRPT$^+$ preempts the update $f$ under transmission 
whenever an update $j$ is generated with size $s_j\le s_f(g_j)$. 
Therefore, under SRPT$^+$, either update $f$ is completely transmitted at time $g_i+s_f(g_i)$, or it gets preempted, and some update $j\ge i$ begins transmission at time $g_j$. 
Thus, the earliest time instant when SRPT$^+$ completes the transmission of update $f$ (which is under transmission at time $g_i$), or an update $j$ (generated at $g_j\ge g_i$) is 
$r_f^+=\min\{g_i+s_f(g_i),\min_{j\ge i}\{g_j+s_j\}\}\le \min_{j\ge i}\{g_j+s_j\}\le r_i^\star$, where the last inequality follows from Lemma \ref{lemma:lb-nu-star}. Since $\nu_i^\star=r_i^\star-g_i$, we get that $r_f^+\le r_i^\star=g_i+\nu_i^\star$. 
Now, consider the following two complementary cases.
\subsubsection{Update $f$ is preempted, and an update $j\ge i$ is completely transmitted at time $r_f^+\le r_i^\star=g_i+\nu_i^\star$.} 

By definition, $w_i^+=b_i^+-g_i\le r_i^+-g_i$. Also, at $r_f^+$, an update $j\ge i$ is completely transmitted. Therefore, $r_i^+=r_f^+$. 
Combining these two facts, we get $w_i^+\le r_f^+-g_i\le r_i^\star-g_i=\nu_i^\star\le 2\nu_i^\star$.

\subsubsection{Update $f$ is not preempted, i.e. the transmission of update $f$ itself completes at time $r_f^+\le r_i^\star=g_i+\nu_i^\star$.}
Let $w_i^+>2\nu_i^\star$, i.e., until time $g_i+2\nu_i^\star$, under SRPT$^+$, only the updates generated strictly before $g_i$ are under transmission. 
With SRPT$^+$, an update $j$ with generation time $g_j<g_f$ can never be under transmission at time $t\ge r_f^+$ (because for such an update $j$, the index $\gamma_j(t)=0$). Thus, the assumption that $w_i^+>2\nu_i^\star$ implies that in interval $(r_f^+,g_i+2\nu_i^\star]$, SRPT$^+$ begins to transmit only those updates
that got generated in interval $(g_f,g_i)$ are under transmission. 

Note that in interval $(r_f^+,g_i+2\nu_i^\star]$, update $i$ is available at the source, and by definition, SRPT$^+$ never idles when a new update (compared to the completely transmitted updates) is available. 
Thus, the fact that SRPT$^+$ does not begin transmitting update $j\ge i$ in interval $(r_f^+,g_i+2\nu_i^\star]$, implies that at each time instant $t\in(r_f^+,g_i+2\nu_i^\star]$, SRPT$^+$ is transmitting some update generated in interval $(g_f,g_i)$.
However, as we show next in Proposition \ref{prop:not-busy}, this cannot be true, contradicting the assumption that $w_i^+>2\nu_i^\star$. Thus, we conclude that $w_i^+\le 2\nu_i^\star$ in this case as well.
\end{proof}
\begin{proposition} \label{prop:not-busy} 
	Under SRPT$^+$, there exists some time instant $t\in(r_f^+,g_i+2\nu_i^\star]$, when none of the updates generated in interval $(g_f,g_i)$ is under transmission.
\end{proposition}
\begin{proof}
	We prove the result using contradiction. With SRPT$^+$, let at each time instant $t\in(r_f^+,g_i+2\nu_i^\star]$, some update $k$ that is generated in interval $(g_f,g_i)$ is under transmission.
	Consider the time instant $r_i^\star=g_i+\nu_i^\star$. Since $r_f^+\le g_i+\nu_i^\star$, we get that $r_i^\star\in(r_f^+,g_i+2\nu_i^\star]$. By definition of $r_i^\star$, there exists some update $j^\star\ge i$ that $\pi^\star$ transmits completely by time $r_i^\star=g_i+\nu_i^\star$. Thus, for update $j^\star$, $g_{j^\star}+s_{j^\star}\le g_i+\nu_i^\star$. 
Since $g_{j^\star}\ge g_i$, this implies that $s_{j^\star}\le\nu_i^\star$. Also, because $s_{j^\star}\ge 0$, we have $g_{j^\star}\le g_i+\nu_i^\star$. Thus, at time $g_i+\nu_i^\star$, the size $s_{j^\star}(g_i+\nu_i^\star)\stackrel{(a)}{\le}s_{j^\star}(g_{j^\star})=s_{j^\star}\le \nu_i^\star$, where $(a)$ follows because the remaining size of an update can never increase. Using this fact, as well as the fact that ${j^\star}\ge i$ (i.e. $g_{j^\star}\ge g_i$), we get the following lower bound on the index \eqref{eq:slope} for update ${j^\star}$ at time $g_i+\nu_i^\star$:
\begin{align} \label{eq:index-j}
	\gamma_{j^\star}(g_i+\nu_i^\star)\stackrel{(a)}{=}\frac{g_{j^\star}-\lambda^+(g_i+\nu_i^\star)}{s_{j^\star}(g_i+\nu_i^\star)}\ge 
	\frac{g_i-\lambda^+(g_i+\nu_i^\star)}{\nu_i^\star},
\end{align}
where in $(a)$, $\lambda^+(g_i+\nu_i^\star)$ denotes the generation time of the latest update that is completely transmitted by SRPT$^+$ until time $g_i+\nu_i^\star$. 

Let $\cK$ denote the set of all updates generated in interval $(g_f,g_i)$) that SRPT$^+$ transmits completely/partially in interval $(g_i+\nu_i^\star,g_i+2\nu_i^\star]$ 
(because $r_i^\star=g_i+\nu_i^\star\in (r_f^+,g_i+2\nu_i^\star]$, the interval $(g_i+\nu_i^\star,g_i+2\nu_i^\star]\subseteq (r_f^+,g_i+2\nu_i^\star]$).
At each time instant $t\in(g_i+\nu_i^\star,g_i+2\nu_i^\star]$, since SRPT$^+$ is transmitting updates from subset $\cK$ despite having update ${j^\star}$ with greater generation time (compared to the updates in $\cK$), 
it follows that the index \eqref{eq:slope} of the updates in $\cK$ that SRPT$^+$ transmits, must be greater than \eqref{eq:index-j}. 
Also, the sum of the sizes of updates in $\cK$ must be at least $\nu_i^\star$ (as the length of interval $(g_i+\nu_i^\star,g_i+2\nu_i^\star]$ is $\nu_i^\star$). Therefore, if all the updates in $\cK$ are completely transmitted, then it must lead to reduction in AoI at least as much as $g_i-\lambda^+(g_i+\nu_i^\star)$. 
However, this cannot be true because the generation time of the updates in $\cK$ is strictly less than $g_i$. Thus, SRPT$^+$ cannot be transmitting updates from subset $\cK$ at each time instant $t\in(g_i+\nu_i^\star,g_i+2\nu_i^\star]\subseteq (r_f^+,g_i+2\nu_i^\star]$. That is, there exists some time instant $t\in(r_f^+,g_i+2\nu_i^\star]$ when none of the updates generated in interval $(g_f,g_i)$ is under transmission (with SRPT$^+$).
\end{proof} 

\subsection{$\sum_{i=1}^{R(t)}\delta_i d_i^+\le 2\sum_{i=1}^{R(t)}\delta_i \nu_i^\star$, where $\delta_i=g_i-g_{i-1}$, 	while $d_i^+=r_i^+-b_i^+$.} \label{app:proof-lemma-SRPT+-2}
\begin{proof}
Recall that $b_i^+$ is the earliest time instant when the transmission of an update $j\ge i$ begins under SRPT$^+$. By definition, $b_i^+$'s may be equal (common) for successively generated updates.
Let the set of updates generated 
at the source until time $t$ (where $t\to\infty$) be partitioned into the largest subsets 
$\cA_1,\cA_2,\cA_3,\cdots$, such that for all updates $i$ within a subset $\cA_n$, $b_i^+$'s are equal.
To prove $\sum_{i=1}^{R(t)}\delta_i d_i^+\le 2\sum_{i=1}^{R(t)}\delta_i \nu_i^\star$, we show that for any subset $\cA_n$, 
\begin{align} \label{eq:temp-A2-1}
	\sum_{i\in\cA_n}\delta_i d_i^+\le 2\sum_{i\in\cA_n}\delta_i \nu_i^\star.
\end{align} 

Consider a subset $\cA_n$ (for some $n$). 
Let the updates in $\cA_n$ be numbered as $1_n,2_n,\cdots,m_n$, in increasing order of their generation time. By definition of $\cA_n$,  
	$b_{1_n}^+=b_{2_n}^+=\cdots=b_{m_n}^+$.
Also, the update whose transmission begins at time $b_{m_n}^+$ is update $m_n$ (latest generated update in the subset $\cA_n$), as shown next. 
\begin{proposition}
	At time $b_{m_n}^+$, SRPT$^+$ begins to transmit update ${m_n}$.
\end{proposition}
\begin{proof}
By definition, $b_i^+$ is the earliest time instant when the transmission of an update $j\ge i$ begins. Since $b_{1_n}^+=b_{2_n}^+=\cdots=b_{m_n}^+$, at time $b_{m_n}^+$, either the transmission of update ${m_n}$, or an update generated after $m_n$ (say, $k_n>m_n$), begins. But if the transmission of $k_n$ begins before update $m_n$, then by definition, $b_{k_n}^+=b_{m_n}^+$, i.e. $k_n\in\cA_n$, which is not true (by definition of $\cA_n$). This implies that at time $b_{m_n}^+$, SRPT$^+$ begins to transmit update $m_n$.
\end{proof}

By definition, once SRPT$^+$ begins to transmit update $m_n$ at time $b_{m_n}^+$, it never transmits any update $j$ generated before $m_n$, at time $t>b_{m_n}^+$ (i.e. after SRPT$^+$ begins to transmit update $m_n$). Therefore, $\forall i\in\cA_n$, the earliest time instant when the transmission of an update $j\ge i$ completes is $r_i^+=r_{m_n}^+$ (the update whose transmission completes at time $r_{m_n}^+$ need not be update $m_n$). Therefore, $d_i^+=r_i^+-b_i^+=r_{m_n}^+-b_{m_n}^+=d_{m_n}^+$, $\forall i\in\cA_n$. Hence, 
\begin{align} \label{eq:size-rel-1}
	\sum_{i\in\cA_n}\delta_i d_i^+= d_{m_n}^+\cdot\sum_{i\in\cA_n}\delta_i.
\end{align}

As discussed in Remark \ref{remark:greedy+}, regardless of whether SRPT$^+$ begins to transmit an update due to larger index \eqref{eq:slope} (when no update is being transmitted) or due to preemption by smaller sized update compared to the 
remaining size of the update under transmission, the time instant when SRPT$^+$ begins to transmit an update, the update's index \eqref{eq:slope} is the largest. Therefore, at time $b_{m_n}^+$ when SRPT$^+$ begins to transmit update $m_n$, the index $\gamma_{m_n}(b_{m_n}^+)$ must be the maximum among all the available updates. Thus, 
$\forall j\in\cA_n=\{1_n,\cdots,m_n\}$, 
\begin{align} \label{eq:rel-slope}
	\frac{g_{m_n}-\lambda^+(b_{m_n}^+)}{s_{m_n}}\ge \frac{g_j-\lambda^+(b_{m_n}^+)}{s_j},
\end{align} 
where $s_j=s_j(b_{m_n}^+)=s_j(g_j)$ is the remaining size of update $j\in\cA_n$ at time $b_{m_n}^+$ under SRPT$^+$, while the superscript $+$ in $\lambda^+(\cdot)$ denotes 
SRPT$^+$.

Without loss of generality, let update $0_n$ be the latest update generated before $g_{1_n}$ (generation time of update $1_n\in\cA_n$). Then $\delta_{1_n}=g_{1_n}-g_{0_n}$, where $g_{0_n}$ denotes the generation time of update $0_n$. 
Recall that $\lambda^+(b_{m_n}^+)$ is equal to the generation time of latest update that has been completely transmitted by SRPT$^+$ until time $b_{m_n}^+$. 
Since update $m_n$ is the only update in $\cA_n$ whose transmission begins until time $b_{m_n}^+$, it follows that $\lambda^+(b_{m_n}^+)\le g_{0_n}$. Therefore, $g_{0_n}-\lambda^+(b_{m_n}^+)=\epsilon\ge 0$.  

Note that $g_j-\lambda^+(b_{m_n}^+)=\epsilon+\sum_{i=1}^j\delta_i$, $\forall j\in\{1_n,\cdots,m_n\}$.
Thus, \eqref{eq:rel-slope} can be written as
\begin{align}
	\frac{\epsilon+\sum_{i=1_n}^{m_n}\delta_i}{s_{m_n}}\ge \frac{\epsilon+\sum_{i=1_n}^j\delta_i }{s_j}, 
\end{align}
which implies that $\forall j\in\cA_n$, 
\begin{align} \label{eq:size-updates}
	s_j&\ge s_{m_n}\cdot\left(\frac{\epsilon+\sum_{i=1_n}^j\delta_i}{\epsilon+\sum_{i=1_n}^{m_n}\delta_i}\right), \nonumber \\
	&\stackrel{(a)}{\ge} s_{m_n}\cdot\left(\frac{\sum_{i=1_n}^j\delta_i}{\sum_{i=1_n}^{m_n}\delta_i}\right),
\end{align}
where we get $(a)$ because $\sum_{i=1_n}^j\delta_i\le \sum_{i=1_n}^{m_n}\delta_i$, and $\epsilon\ge 0$.


Now, we connect the decisions of $\pi^\star$ (an optimal offline policy) with those of SRPT$^+$, over each subset $\cA_n$.
\begin{proposition} \label{prop:size-d-updates}
	For all updates $j\in\cA_n=\{1_n,\cdots,m_n\}$, 
	\begin{align} \label{eq:size-d-updates}
		\nu_j^\star \ge  d_{m_n}^+\cdot\left(\frac{\sum_{i=1_n}^j\delta_i}{\sum_{i=1_n}^m\delta_i}\right).
	\end{align}
\end{proposition}
\begin{proof}
Let $\theta_j=\sum_{i=1_n}^j\delta_i/\sum_{i=1_n}^{m_n}\delta_i$, $\forall j\in\cA_n=\{1_n,\cdots,m_n\}$ (by definition, $\theta_j\le 1$). From Lemma \ref{lemma:lb-nu-star}, we know that for any update $j$, $r_j^\star\ge \min_{i\ge j}\{g_i+s_i\}$, which implies 
\begin{align} \label{eq:r-star}
	r_j^\star&\ge\min\{\min_{i\in\{j,\cdots,m_n\}}\{g_i+s_i\},\min_{i>m_n}\{g_i+s_i\}\}, \nonumber \\
	&\stackrel{(a)}{\ge} \min\{\min_{i\in\{j,\cdots,{m_n}\}}\{g_i+s_{m_n}\cdot \theta_i\},\min_{i>m_n}\{g_i+s_i\}\},
\end{align}
where 
we get $(a)$ because for any update $i\in\cA_n$, $s_i\ge s_{m_n}\cdot \theta_i$ (from \eqref{eq:size-updates}).

Since $\nu_j^\star=r_j^\star-g_j\ge r_j^\star-b_j^\star=d_j^\star$, on subtracting $g_j$ from both sides of \eqref{eq:r-star} and denoting $\phi_{ij}=g_i-g_j$ $\forall i\ge j$, we get 
\begin{align} \label{eq:nu-j-star}
	\nu_j^\star&\ge \min\{\min_{i\in\{j,\cdots,m_n\}}\{\phi_{ij}+s_{m_n}\cdot \theta_i\},\min_{i>m_n}\{\phi_{ij}+s_i\}\}, \nonumber \\
	&\stackrel{(a)}{\ge} \min\{s_{m_n}\cdot\theta_j,\ \ \min_{i>{m_n}}\{\phi_{ij}+s_i\}\},
\end{align}
where we get $(a)$ because $\phi_{ij}=g_i-g_j\ge 0$, $\forall i\ge j$.

Note that under SRPT$^+$, if transmission of update $m_n$ (which starts at time $b_{m_n}^+$) completes, then $r_{m_n}^+=b_{m_n}^++s_{m_n}$. Else, if update $m_n$ gets preempted, then $r_{m_n}^+=\min_{i>m_n}\{g_i+s_i\}$. Combining these two facts, we get
$r_{m_n}^+=\min\{b_{m_n}^++s_{m_n},\min_{i>m_n}\{g_i+s_i\}\}$. Thus, $d_{m_n}^+=r_{m_n}^+-b_{m_n}^+=\min\{s_{m_n},\min_{i>m_n}\{(g_i-b_{m_n}^+)+s_i\}\}$. Since $g_j\le b_j^+=b_{m_n}^+$ $\forall j\in\cA_n$, 
we have $g_i-b_{m_n}^+\le g_i-g_j=\phi_{ij}$. Therefore,
	$d_{m_n}^+\le \min\{s_{m_n}, \min_{i>m_n}\{\phi_{ij}+s_i\}\}$.
Multiplying both sides by $\theta_j$, we get $d_{m_n}^+\cdot\theta_j \le\min\{s_{m_n}\cdot \theta_j,\ \  \min_{i>m_n}\{\phi_{ij}+s_i\}\cdot\theta_j\}$. Since $\theta_j\le 1$ (by definition), we get
\begin{align} \label{eq:d-m-G}
	d_{m_n}^+\cdot\theta_j 
	&\le \min\{s_{m_n}\cdot \theta_j,\ \ \min_{i>m_n}\{\phi_{ij}+s_i\}\}.
\end{align} 

Comparing \eqref{eq:nu-j-star} and \eqref{eq:d-m-G}, we get \eqref{eq:size-d-updates}.
\end{proof}

Multiplying both sides of \eqref{eq:size-d-updates} by $\delta_j$ and summing over $j\in\cA_n=\{1_n,\cdots,m_n\}$, we get
\begin{align} \label{eq:size-rel-2}
	\sum_{j\in\cA_n}\delta_j \nu_j^\star=\sum_{j=1_n}^{m_n}\delta_j \nu_j^\star &\ge \frac{d_{m_n}^+}{\sum_{i=1_n}^{m_n}\delta_i} \cdot\sum_{j=1_n}^{m_n}\left(\delta_j\sum_{i=1_n}^j\delta_i\right), \nonumber \\
	&\ge \frac{d_{m_n}^+}{\sum_{i=1_n}^{m_n}\delta_i} \cdot\frac{\left(\sum_{i=1_n}^{m_n}\delta_i\right)^2}{2}, \nonumber \\
	&=\frac{d_{m_n}^+\cdot\sum_{i=1_n}^{m_n}\delta_i}{2}, \nonumber \\
	&=\frac{1}{2}\left(d_{m_n}^+\cdot\sum_{i\in\cA_n}\delta_i\right).
\end{align}

From \eqref{eq:size-rel-1} and \eqref{eq:size-rel-2}, we get \eqref{eq:temp-A2-1}, and summing both sides of \eqref{eq:temp-A2-1} over all $n$ (i.e. subsets $\cA_n$), we get $\sum_{i=1}^{R(t)}\delta_i d_i^+\le 2\sum_{i=1}^{R(t)}\delta_i \nu_i^\star$.
\end{proof}

\section{Proof of Lemma \ref{lemma:LCFS+properties}} \label{app:proof-lemma-LCFS+properties}
\subsection{For any update $i$, $w_i^L\le \nu_i^{\min}\le \nu_i^\star$.} 

\begin{proof}
	Consider the following two complementary cases.
	\begin{enumerate}
		\item At the generation time $g_i$ of (any) update $i$, no update is under transmission: Under SRPT$^L$, the transmission of update $i$ begins immediately. Therefore, $w_i^L=0\stackrel{(a)}{\le} \nu_i^{\min}\stackrel{(b)}{\le} \nu_i^\star$, where $(a)$ and $(b)$ follows from Lemma \ref{lemma:lb-nu-star}.
		\item An update $f$ is already under transmission at time $g_i$: 
		From Lemma \ref{lemma:lb-nu-star}, we have 
		\begin{align} \label{eq:temp-2}
			r_i^\star\ge\min_{j\ge i}\{g_j+s_j\}=g_i+\nu_i^{\min}\ge\min\{g_i+s_f(g_i),\min_{j\ge i}\{g_j+s_j\}\}\stackrel{(a)}{=}r_f^L,
		\end{align}
		where $r_f^L$ is the earliest time instant when SRPT$^L$ completely transmits an update $j\ge f$, and we get $(a)$ because SRPT$^L$ preempts the update under transmission whenever a new update with smaller remaining size is generated. Subtracting $g_i$ from both sides of the relations in \eqref{eq:temp-2}, we get 
		\begin{align} \label{eq:temp-3}
			\nu_i^\star=r_i^\star-g_i\ge \nu_i^{\min} \ge r_f^L-g_i.
		\end{align}
		Note that if update $f$ is preempted (at time $t\ge g_i$), then the transmission of an update $j\ge i$ begins before $r_f^L$, else the transmission of an update $j\ge i$ begins at most at $r_f^L$ (when transmission of update $f$ completes, as SRPT$^L$ begins to transmit the latest generated update). Therefore, $b_i^L\le r_f^L$, which implies $w_i^L=b_i^L-g_i\le r_f^L-g_i\stackrel{(a)}{\le} \nu_i^{\min}\stackrel{(b)}{\le} \nu_i^\star$, where $(a)$ and $(b)$ follows from \eqref{eq:temp-3}.
	\end{enumerate}
	
	From the above two cases we get that for any update $i$, $w_i^L\le \nu_i^{\min}\le \nu_i^\star$.
\end{proof}

\subsection{At the time instant $t$ when SRPT$^L$ begins to transmit an update $i$, its generation time $g_i$ is the latest among all updates generated until time $t$. That is, SRPT$^L$ either never begins to transmit an update $i$, or begins to transmit it at some time $t\in[g_i,g_{i+1})$.} 

\begin{proof}
	Recall that SRPT$^L$ begins to transmit an update $i$ either a) by preempting another update at time $g_i$, or b) at the earliest time instant $t$ when update $i$ is the latest generated update and no update is under transmission (which implies that $t\in[g_i,g_{i+1})$). Hence, if SRPT$^L$ begins to transmit an update $i$ at time $t$, then $t\in[g_i,g_{i+1})$.
\end{proof}

\subsection{If SRPT$^L$ begins to transmit an update $i$ (at time $b_i^L$), then $r_i^L=\min\{b_i^L+s_i,\min_{j\ge i+1}\{g_j+s_j\}\}$. Additionally, for such an update $i$,  $d_i^L=r_i^L-b_i^L\le \nu_i^{\min}\le \nu_i^\star$.}

\begin{proof}
	Note that SRPT$^L$ preempts the update under transmission whenever a new update is generated with size at most the remaining size of the update under transmission.
	Therefore, if SRPT$^L$ begins to transmit an update $i$ at time $b_i^L\in[g_i,g_{i+1})$, then the earliest time instant when the transmission of an update $j\ge i$ completes is $r_i^L=\min\{b_i^L+s_i,\min_{j\ge i+1}\{g_j+s_j\}\}$.
	Further, by definition, $d_i^L=r_i^L-b_i^L=\min\{s_i,\min_{j\ge i+1}\{(g_j-b_i^L)+s_j\}\}\stackrel{(a)}{\le} \min\{s_i,\min_{j\ge i+1}\{(g_j-g_i)+s_j\}\}=\min_{j\ge i}\{g_j+s_j\}-g_i\stackrel{(b)}{=}\nu_i^{\min}\le \nu_i^\star$, where $(a)$ follows because $b_i^L\ge g_i$, and $(b)$ follows from the definition of $\nu_i^{\min}$.
\end{proof}

\section{An example to illustrate the properties of SRPT$^L$} \label{app:example-LCFS}

\begin{example} \label{ex:SRPT-L}
	At time $t=0$, let the AoI $\Delta(0)=0$. 
	In time interval $[0,2]$, the update generation sequence $\cI=\{(0,1.45),(0.25,1.25),(0.75,1),(1,0.5),(1.25,0.3), (1.8,0.1)\}$, and the updates $i$ (i.e. $(g_i,s_i)$) in $\cI$ are numbered (ordered) in increasing order of their generation time $g_i$. 
	
	At time $t=0$, since no update is under transmission, SRPT$^L$ begins to transmit the latest update, i.e. update $1$. Subsequently, at time $t=0.25$, update $2$ is generated with size $s_2=1.25>(1.45-0.25)=s_1(0.25)$. Therefore, update $1$ remains under transmission. At time $t=0.75$, update $3$ is generated with size $s_3=0.75>(1.45-0.75)=s_1(0.75)$, and hence update $3$ does not preempt update $1$ (i.e. update $1$ remains under transmission). Similarly, at time $t=1$ and $t=1.25$ respectively, updates $4$ and $5$ are generated, and none of them preempt update $1$ because at time $t=1$ and $t=1.25$, the remaining size of update $1$ is smaller than the sizes of updates $4$ and $5$ respectively. Thus, the transmission of update $1$ completes at time $t=1.45$.
	
	At time $t=1.45$, no update is under transmission, and the latest update i.e. update $5$ is available for transmission. Therefore, SRPT$^L$ begins to transmit update $5$ at time $t=1.45$, and transmits it completely by time $t=1.45+s_5=1.45+0.3=1.75$. At time $t=1.75$, since latest update is completely transmitted, SRPT$^L$ idles until time $t=1.8$ when update $6$ is generated. In time interval $(1.8,1.9]$, SRPT$^L$  transmits update $6$ completely. Thereafter, in absence of a new update, SRPT$^L$ idles until time $t=2$. 
	
	In summary, in interval $[0,2]$, SRPT$^L$ completely transmits update $1$ (i.e. $(0,1.45)$) over interval $(0,1.45]$, update $5$ (i.e. $(1.25,0.3)$) over interval $(1.45,1.75]$ and update $6$ (i.e. $(1.8,0.1)$) over interval $(1.8,1.9]$. Thus, the AoI increases linearly from $0$ at time $t=0$ to $1.45$ at time $t=1.45$. At time $t=1.45$, the generation time of latest completely transmitted update is $\lambda(1.45)=g_1=0$, which implies that the AoI $\Delta(1.45)=1.45-\lambda(1.45)=1.45$. Thus, there is no reduction in AoI because of complete transmission of update $1$. Consequently, the AoI increases linearly from $1.45$ at time $t=1.45$ to $1.75$ at time $t=1.75$, and then instantaneously decreases to $1.75-g_5=1.75-1.25=0.5$ after update $5$ is completely transmitted. Then, AoI increases linearly from $0.5$ at time $t=1.75$ to $0.65$ at time $t=1.9$, and then decreases instantaneously to $1.9-g_6=1.9-1.8=0.1$ after update $6$ is completely transmitted. Thereafter, AoI again starts to increase linearly. 
	
	Thus, we get that the average AoI (i.e. $\av{AoI}$) for SRPT$^L$ in interval $[0,2]$ is $1.6325$ time units. Recall from Example \ref{ex:SRPT-not-OPT} that for the same update generation sequence $\cI$, the $\av{AoI}$ for SRPT$^+$ is $1.395$ time units, whereas the $\av{AoI}$ for an optimal offline policy $\pi^\star$ is $1.3825$ time units.
\end{example}

\section{Proof of Theorem \ref{thm:LCFS}} \label{app:proof-thm-LCFS}

\begin{proof}
	By definition, 
	$b_i^L\ge g_i$ is the earliest time instant when SRPT$^L$ begins to transmit an update $j\ge i$. 
	We partition the set of all updates generated at the source into sets $G_=^L$ and $G_>^L$, based on $b_i^L$'s. In particular, we define $G_=^L$ as the subset of all updates $i$ such that $b_i^L=g_i$ (i.e., under SRPT$^L$, the transmission of update $i$ begins immediately at its generation time), and $G_>^L$ is the subset of all updates $i$ such that $b_i^L>g_i$. Then, the $\av{AoI}$ \eqref{eq:aaoi-general} for SRPT$^L$ can be written as 
	\begin{align} \label{eq:aaoi-LCFS-general}
		\aaoi_{\text{SRPT}^L}&=\lim_{t\to\infty}\frac{1}{t}\left[\sum_{i=1}^{R(t)}\frac{\delta_{i}^2}{2}+\sum_{i=1}^{R(t)}\delta_i w_i^L+\sum_{i=1}^{R(t)}\delta_i d_i^L\right], \nonumber \\
		&=\lim_{t\to\infty}\frac{1}{t}\left[\sum_{i=1}^{R(t)}\frac{\delta_{i}^2}{2}+\left(\sum_{i\in G_=^L}\delta_i w_i^L+\sum_{i\in G_>^L}\delta_i w_i^L\right)+\left(\sum_{i\in G_=^L}\delta_i d_i^L+\sum_{i\in G_>^L}\delta_i d_i^L\right)\right], \nonumber \\
		&\stackrel{(a)}{=}\lim_{t\to\infty}\frac{1}{t}\left[\sum_{i=1}^{R(t)}\frac{\delta_{i}^2}{2}+\sum_{i\in G_>^L}\delta_i w_i^L+\sum_{i\in G_=^L}\delta_i d_i^L+\sum_{i\in G_>^L}\delta_i d_i^L\right],
	\end{align} 
	where we get $(a)$ because by definition of $G_=^L$, for any $i\in G_=^L$, $b_i^L=g_i$, which implies $w_i^L=b_i^L-g_i=0$.
	
	From property $1$ in Lemma \ref{lemma:LCFS+properties}, we get that for each $i\in G_>^L$, $w_i^L\le \nu_i^\star$, which implies 
	\begin{align} \label{eq:prop1-lemma}
		\sum_{i\in G_>^L}\delta_i w_i^L\le \sum_{i\in G_>^L}\delta_i \nu_i^\star.
	\end{align}
	Further, by definition, of $G_=^L$, SRPT$^L$ begins to transmit each update in $G_=^L$ at generation time. Therefore, from property $3$
	of Lemma \ref{lemma:LCFS+properties}, we get that for each update $i\in G_=^L$, $d_i^L\le \nu_i^\star$, which implies 
	\begin{align} \label{eq:prop2-lemma}
		\sum_{i\in G_=^L}\delta_i d_i^L\le \sum_{i\in G_=^L}\delta_i \nu_i^\star.
	\end{align}
	
	Substituting \eqref{eq:prop1-lemma} and \eqref{eq:prop2-lemma} in \eqref{eq:aaoi-LCFS-general}, we get 
	\begin{align}
		\aaoi_{\text{SRPT}^L}&\le \lim_{t\to\infty}\frac{1}{t}\left[\sum_{i=1}^{R(t)}\frac{\delta_{i}^2}{2}+\sum_{i\in G_>^L}\delta_i \nu_i^\star+\sum_{i\in G_=^L}\delta_i \nu_i^\star\right]+\lim_{t\to\infty}\frac{1}{t}\sum_{i\in G_>^L}\delta_i d_i^L, \nonumber \\
		&=\lim_{t\to\infty}\frac{1}{t}\left[\sum_{i=1}^{R(t)}\frac{\delta_{i}^2}{2}+\sum_{i=1}^{R(t)}\delta_i \nu_i^\star\right]+\lim_{t\to\infty}\frac{1}{t}\sum_{i\in G_>^L}\delta_i d_i^L, \nonumber \\
		&\stackrel{(a)}{\le}\lim_{t\to\infty}\frac{1}{t}\left[\sum_{i=1}^{R(t)}\frac{\delta_{i}^2}{2}+\sum_{i=1}^{R(t)}\delta_i \nu_i^\star\right]+28\aaoi_{\pi^\star}, \nonumber \\
		&\stackrel{(b)}{=}29\aaoi_{\pi^\star},
	\end{align}
	where $(a)$ follows from Lemma \ref{lemma:LCFS-technical-lemma} (discussed below), and $(b)$ follows from \eqref{eq:aaoi-opt}.
\end{proof}

\begin{lemma} \label{lemma:LCFS-technical-lemma} 
	$\underset{t\to\infty}{\lim}\frac{1}{t}\sum_{i\in G_>^L}\delta_i d_i^L\le 28\aaoi_{\pi^\star}=\underset{t\to\infty}{\lim}\frac{28}{t}\int_{0}^t\Delta_{\pi^\star}(\tau)d\tau$.
\end{lemma}
\begin{proof}[Proof Sketch (for detailed proof, see Appendix \ref{app:proof-lemma-LCFS-technical-lemma})]
	The basic idea is to partition the set $G_>^L$ into subsets $\cA_n$'s such that corresponding to each subset $\cA_n$, there exists a time-interval $\chi_n$, such that 
	\begin{align} \label{eq:sketch-LCFS}
		\sum_{i\in\cA_n}\delta_i d_i^L\le c_n\cdot\int_{\chi_n}\Delta_{\text{SRPT}^L}(\tau)d\tau,
	\end{align}
	for some constant $c_n$. We consider $\cA_n$'s, such that the generation time $b_i^L$ for all updates $i\in\cA_n$ (for fixed $n$) are equal. Correspondingly, we define $\chi_n$ as the interval between the generation time of the latest update generated before the updates in $\cA_n$, and the generation time of the earliest update whose transmission begins under LCFS$^L$ after all the updates in $\cA_n$ are generated.
	
	Since $\cA_n$'s partition $G_>^L$, summing both sides of \eqref{eq:sketch-LCFS} over all $n$, we get 
	\begin{align} \label{eq:sketch-LCFS-2}
		\underset{t\to\infty}{\lim}\frac{1}{t}\sum_{i\in G_>^L}\delta_i d_i^L\le \underset{t\to\infty}{\lim}\frac{1}{t}\sum_{\forall n}\int_{\chi_n}c_n\Delta_{\pi^\star}(\tau)d\tau.
	\end{align} 
	
	Note that the L.H.S. of \eqref{eq:sketch-LCFS-2} is equal to the L.H.S. of Lemma \ref{lemma:LCFS-technical-lemma}. Thus, to prove Lemma \ref{lemma:LCFS-technical-lemma}, we show that the R.H.S. of Lemma \ref{lemma:LCFS-technical-lemma} is an upper bound on the R.H.S. of \eqref{eq:sketch-LCFS-2}.
	To show this, we classify the intervals $\chi_n$'s into disjoint groups such that $c_n$'s are equal for each interval $\chi_n$ within a group. For the choice of $\cA_n$'s in our proof, we identify two groups i.e. $G_>^L(1)$ and $G_>^L(2)$, such that $c_n=8$ for each $n\in G_>^L(1)$, and $c_n=10$ for each $n\in G_>^L(2)$. Thus, \eqref{eq:sketch-LCFS-2} simplifies to 
	\begin{align} \label{eq:sketch-LCFS-3}
		\underset{t\to\infty}{\lim}\frac{1}{t}\sum_{i\in G_>^L}\delta_i d_i^L\le \underset{t\to\infty}{\lim}\frac{8}{t}\sum_{n\in G_>^L(1)}\int_{\chi_n}\Delta_{\pi^\star}(\tau)d\tau+\underset{t\to\infty}{\lim}\frac{10}{t}\sum_{n\in G_>^L(2)}\int_{\chi_n}\Delta_{\pi^\star}(\tau)d\tau.
	\end{align} 
	
	Further, we show that $\forall n\in G_>^L(1)$, $\chi_n$'s are disjoint, and $\cup_{n\in G_>^L(1)}\chi_n\subseteq [0,t]$. Therefore, 
	\begin{align} \label{eq:sketch-LCFS-4}
		\lim_{t\to\infty}\frac{8}{t}\sum_{n\in G_>^L(1)}\int_{\chi_n}\Delta_{\pi^\star}(\tau)d\tau\le \lim_{t\to\infty}\frac{8}{t}\int_{0}^t\Delta_{\pi^\star}(\tau)d\tau= 8\aaoi_{\pi^\star}
	\end{align}
	However, $\chi_n$'s for $n\in G_>^L(1)$ may not be disjoint, but we show that alternate intervals are disjoint, i.e. intervals $\chi_1, \chi_3,\chi_5,\cdots$ are mutually disjoint, and intervals $\chi_2, \chi_4,\chi_6,\cdots$ are mutually disjoint. Hence, $\underset{t\to\infty}{\lim}\frac{10}{t}\sum_{n\in G_>^L(2)}\int_{\chi_n}\Delta_{\pi^\star}(\tau)d\tau$
	\begin{align} \label{eq:sketch-LCFS-5}
		=&\underset{t\to\infty}{\lim}\frac{10}{t}\sum_{n\in G_>^L(2),n\text{ odd}}\int_{\chi_n}\Delta_{\pi^\star}(\tau)d\tau +\underset{t\to\infty}{\lim}\frac{10}{t}\sum_{n\in G_>^L(2),n\text{ even}}\int_{\chi_n}\Delta_{\pi^\star}(\tau)d\tau, \nonumber \\
		\le& \underset{t\to\infty}{\lim}\frac{10}{t}\int_{0}^t\Delta_{\pi^\star}(\tau)d\tau +\underset{t\to\infty}{\lim}\frac{10}{t}\int_{0}^t\Delta_{\pi^\star}(\tau)d\tau, \nonumber \\
		=&10\aaoi_{\pi^\star}+10\aaoi_{\pi^\star}.
	\end{align}
	Substituting \eqref{eq:sketch-LCFS-4} and \eqref{eq:sketch-LCFS-5} into \eqref{eq:sketch-LCFS-3}, we get Lemma \ref{lemma:LCFS-technical-lemma}. 
\end{proof}

\section{Proof of Lemma \ref{lemma:LCFS-technical-lemma}} \label{app:proof-lemma-LCFS-technical-lemma}

\begin{proof}
Partition the set $G_>^L$ into largest subsets $\cA_1,\cA_2,\cdots$, such that a) for any subset $\cA_n$, 
$b_i^L$ 
is same for all updates $i\in\cA_n$, and b) for update $i\in\cA_n$ and $j\in\cA_{n'}$, $b_i^L< b_j^L$ whenever $n< n'$. Further, let $\cB$ denote the set of all subsets $\cA_n$'s that partition $G_>^L$ ($\cA_n$'s are sets of updates, while $\cB$ is a set of subsets $\cA_n$'s). 
\begin{remark}
	The second condition in the definition of $\cA_n$'s implies that the updates in $\cA_n$ are generated before the updates in $\cA_{n+1}$, $\forall n$.
\end{remark}
\begin{remark}
	Note that the defintion of $\cA_n$'s that we use in this section (and onwards) is different from that in Appendix \ref{app:proof-lemma-SRPT+-2} in one aspect. In Appendix \ref{app:proof-lemma-SRPT+-2}, $A_n$'s partition the set of all updates generated at the source, whereas in this section $\cA_n$'s partition the subset of updates $G_>^L$.
\end{remark}

Since the subsets $\cA_n\in\cB$ partition $\cG_>^L$, we get
\begin{align} \label{eq:split-1}
	\lim_{t\to\infty}\frac{1}{t}\sum_{i\in G_{>}^{L}}\delta_i d_{i}^{L}=\lim_{t\to\infty}\frac{1}{t}\sum_{\cA_n\in\cB}\sum_{i\in \cA_n}\delta_i d_{i}^{L}.
\end{align}

Consider a subset $\cA_n\in\cB$ (for some $n$). Without loss of generality, let the 
updates in $\cA_n$ be numbered (ordered) as $1_n,2_n,\cdots,m_n$, in increasing order of their generation times.
Also, let the latest update generated before update $1_n$ be $0_n$. 
Note that 
$0_n\not\in\cA_n$. 

By definition of $\cA_n$, $b_{1_n}^L=\cdots=b_{m_n}^L$. Additionally, since SRPT$^L$ never transmits an older update after it begins to transmit a new update, we have 
\begin{align} \label{eq:r-i-n-equal}
	r_{1_n}^L=\cdots=r_{m_n}^L.
\end{align}
Combining these two facts, we get $r_{1_n}^L-b_{1_n}^L=\cdots=r_{m_n}^L-b_{m_n}^L$, i.e. $d_{1_n}^L=\cdots=d_{m_n}^L$. Thus, 
\begin{align} \label{eq:lhs-1}
	\sum_{i\in\cA_n}\delta_i d_i^L=\left(\sum_{i\in\cA_n}\delta_i\right)d_{m_n}^L=(g_{m_n}-g_{0_n})d_{m_n}^L.
\end{align}


Let $k_n\ge m_n$ denote the update that SRPT$^L$ begins to transmit at time $b_{m_n}^L$. Then $d_{m_n}^L=d_{k_n}^L$. Using Lemma \ref{lemma:LCFS+properties}, we get $d_{m_n}^L=d_{k_n}^L\le \nu_{k_n}^{\min}$. Thus, upper-bounding $d_{m_n}^L$ in \eqref{eq:lhs-1} by $\nu_{k_n}^{\min}$, we get
\begin{align} \label{eq:lhs-3}
	\sum_{i\in\cA_n}\delta_i d_i^L\le \left(\sum_{i\in\cA_n}\delta_i\right)\nu_{k_n}^{\min}= (g_{m_n}-g_{0_n})\nu_{k_n}^{\min}.
\end{align}

\begin{remark} \label{remark:k_n}
	Note that $b_{k_n}^L=b_{m_n}^L$ 
	(by definition of $k_n$). Therefore, if $k_n\in G_>^L$, then it must be that $k_n=m_n$ (follows from the definition of $\cA_n$'s, and the fact that $k_n\ge m_n$). Otherwise (i.e. if $k_n\in G_=^L$), $k_n$ must be the earliest generated update after $m_n$, 
	else, the updates 
	$m_n+1,\cdots,k_n-1$ 
	would also lie in $\cA_n$, which is not possible (by definition of $m_n$).
\end{remark}
	
	Next, we partition $\cB$ (the set of subsets $\cA_n$'s) as follows:  
	\begin{align} \label{eq:G1-G2}
		G_>^L(1)=\{\cA_n\in \cB|g_{m_n}-g_{0_n}\ge \nu_{k_n}^{\min}\},\text{ and } G_>^L(2)=\{\cA_n\in \cB|g_{m_n}-g_{0_n}<\nu_{k_n}^{\min}\}.
	\end{align}
	Then, we can write 
	$\lim_{t\to\infty}\frac{1}{t}\sum_{\cA_n\in\cB}\sum_{i\in\cA_n}\delta_i d_i^L$
	\begin{align} \label{eq:split-2}
		&=\lim_{t\to\infty}\frac{1}{t}\sum_{\cA_n\in G_>^L(1)}\sum_{i\in\cA_n}\delta_i d_i^L+\lim_{t\to\infty}\frac{1}{t}\sum_{\cA_n\in G_>^L(2)}\sum_{i\in\cA_n}\delta_i d_i^L, \nonumber \\
		&\stackrel{(a)}{\le} \lim_{t\to\infty}\frac{1}{t}\sum_{\cA_n\in G_>^L(1)}(g_{m_n}-g_{0_n}) \nu_{k_n}^{\min}+\lim_{t\to\infty}\frac{1}{t}\sum_{\cA_n\in G_>^L(2)}(g_{m_n}-g_{0_n}) \nu_{k_n}^{\min},
	\end{align}
where $(a)$ follows from \eqref{eq:lhs-3}.
		
	\begin{proposition} \label{prop:split}
i) $\underset{t\to\infty}{\lim}\frac{1}{t}\sum_{\cA_n\in G_>^L(1)}(g_{m_n}-g_{0_n}) \nu_{k_n}^{\min}\le 8\aaoi_{\pi^\star}$, 
ii) $\underset{t\to\infty}{\lim}\frac{1}{t}\sum_{\cA_n\in G_>^L(2)}(g_{m_n}-g_{0_n}) \nu_{k_n}^{\min}\le 20\aaoi_{\pi^\star}$.
	\end{proposition}
\begin{proof}
	See Appendix \ref{app:proof-prop-split}.
\end{proof} 

Combining \eqref{eq:split-1}, \eqref{eq:split-2} and Proposition \ref{prop:split}, we get $\lim_{t\to\infty}\frac{1}{t}\sum_{i\in G_{>}^{L}}\delta_i d_{i}^{L}\le 28\aaoi_{\pi^\star}$.
\end{proof}

\section{Proof of Proposition \ref{prop:split}} \label{app:proof-prop-split}
\subsection{$\lim_{t\to\infty}\frac{1}{t}\sum_{\cA_n\in G_>^L(1)}(g_{m_n}-g_{0_n}) \nu_{k_n}^{\min}\le 8\aaoi_{\pi^\star}$.}

\begin{proof}
Consider the following general result.
			\begin{proposition} \label{prop:nu-0-star-ub}
				For each subset $\cA_n$, 
				i) SRPT$^L$ begins to transmit update $0_n$, and ii) $g_{m_n}-g_{0_n}\le 2\nu_{0_n}^{\min}$.
			\end{proposition}
\begin{proof}
	Since the subsets $\cA_1,\cA_2,\cdots$ partition $G_>^L$, an update may not lie in $\cA_n$ either if it does not lie in $G_>^L$ (i.e. the update lies in $G_=^L$), or it lies in subset $\cA_{n'}$ for some $n'\ne n$. 
	Because update $0_n\not\in\cA_n$ and it is the latest generated update before the updates in $\cA_n$ (by definition), either $0_n\in\cA_{n-1}$, which implies $b_{0_n}^L<b_{1_n}^L$, or ${0_n}\in G_=^L$, which implies $b_{0_n}^L=g_{0_n}<g_{1_n}<b_{1_n}^L$. Thus, in both cases $b_{0_n}^L<b_{1_n}^L$, which implies $b_{0_n}^L<b_{1_n}^L=\cdots=b_{m_n}^L$.
	
	Recall that $b_{0_n}^L$ and $b_{1_n}^L$ are respectively the earliest time instants when under SRPT$^L$, the transmission of an update $j\ge 0_n$ and $j'\ge 1_n$ begins. Since updates $0_n$ and $1_n$ are consecutively generated, and $b_{0_n}^L<b_{1_n}^L$, we get that at $b_{0_n}^L$, SRPT$^L$ begins to transmit update $0_n$.  
	Moreover, since SRPT$^L$ begins to transmit update $0_n$, we get 
	\begin{align} \label{eq:ub-nu-0-L}
		\nu_{0_n}^L=r_{0_n}^L-g_{0_n}=w_{0_n}^L+d_{0_n}^L \stackrel{(a)}{\le} 2\nu_{0_n}^{\min}\stackrel{(b)}{\le} 2\nu_{0_n}^\star,
	\end{align}
where $(a)$ and $(b)$ follows  Lemma \ref{lemma:LCFS+properties}. 
	
	Note that whenever transmission of an update completes and there is a new update to transmit, SRPT$^L$ begins to transmit it immediately. Therefore, if the earliest time when the transmission of an update $j\ge 0_n$ completes is $r_{0_n}^L<g_{m_n}$, then the transmission of some update $j'\in\{1_n,\cdots,m_n-1\}$ will begin before time $g_{m_n}$. But this cannot be true because $b_{1_n}^L=\cdots=b_{m_n}^L$. Hence, $r_{0_n}^L\ge g_{m_n}$, which implies that $r_{0_n}^L-g_{0_n}\ge g_{m_n}-g_{0_n}$. Comparing this result with \eqref{eq:ub-nu-0-L}, we get $g_{m_n}-g_{0_n}\le r_{0_n}^L-g_{0_n}\le 2\nu_{0_n}^{\min}$. 
\end{proof}

From Proposition \ref{prop:nu-0-star-ub} and Remark \ref{remark:k_n}, it follows that $0_n$ and $k_n\ge m_n$ are two successive updates that SRPT$^L$ begins to transmit. Therefore, the intervals $(g_{0_n},g_{m_n}]$, $\forall \cA_n\in\cB$, are subsets of the time intervals between the generation time of two successive updates that SRPT$^L$ begins to transmit. Hence, the intervals $(g_{0_n},g_{m_n}]$, $\forall \cA_n\in\cB$  are disjoint, which implies 
\begin{align} \label{eq:break-integral}
	\aaoi_{\pi^\star}=\lim_{t\to\infty}\frac{1}{t}\int_{0}^{t}\Delta_{\pi^\star}(\tau)d\tau &\ge \lim_{t\to\infty}\frac{1}{t}\sum_{\cA_n\in \cB}\int_{g_{0_n}}^{g_{m_n}}\Delta_{\pi^\star}(\tau)d\tau, \nonumber \\ 
	&\stackrel{(a)}{\ge} \lim_{t\to\infty}\frac{1}{t}\sum_{\cA_n\in G_>^L(1)}\int_{g_{0_n}}^{g_{m_n}}\Delta_{\pi^\star}(\tau)d\tau,
\end{align}
where $(a)$ follows because $G_>^L(1)$ is a subset of $\cB$ (from \eqref{eq:G1-G2}), and summands in \eqref{eq:break-integral} are non-negative (by definition).

From Lemma \ref{lemma:lb-nu-star}, we know that the earliest time instant by which an offline optimal policy $\pi^\star$ can completely transmit an update $j\ge 0_n$ is at least $g_{0_n}+\nu_{0_n}^{\min}$. Thus, with $\pi^\star$, at any time $t\in(g_{0_n},g_{0_n}+\nu_{0_n}^{\min}]$, the generation time of the latest completely transmitted update is $\lambda^\star(t)\le g_{0_n}$, which implies that the AoI of $\pi^\star$ at time $t\in(g_{0_n},g_{0_n}+\nu_{0_n}^{\min}]$ is $\Delta_{\pi^\star}(t)= t-\lambda^\star(t)\ge t-g_{0_n}$. Hence, 

\begin{align} \label{eq:LCFS-proof-p1}
	\int_{g_{0_n}}^{g_{m_n}}\Delta_{\pi^\star}(\tau)d\tau&\ge\int_{g_{0_n}}^{\min\{g_{m_n},\ \ g_{0_n}+\nu_{0_n}^{\min}\}}\Delta_{\pi^\star}(\tau)d\tau, \nonumber \\
	&\ge \frac{\min\{(g_{m_n}-g_{0_n})^2,(\nu_{0_n}^{\min})^2\}}{2}, \nonumber \\
	&\stackrel{(a)}{\ge} \frac{(g_{m_n}-g_{0_n})^2}{8}, \nonumber \\
	&\stackrel{(b)}{\ge} \frac{(g_{m_n}-g_{0_n})\nu_{k_n}^{\min}}{8}, 
\end{align}
where $(a)$ follows from Proposition \ref{prop:nu-0-star-ub}, and $(b)$ is true for $\cA_n\in G_>^L(1)$ (from \eqref{eq:G1-G2}). 

Comparing \eqref{eq:break-integral} and \eqref{eq:LCFS-proof-p1}, 
we get 
$$	\lim_{t\to\infty}\frac{1}{t}\sum_{\cA_n\in G_>^L(1)}(g_{m_n}-g_{0_n}) \nu_{k_n}^{\min}\le \lim_{t\to\infty}\frac{8}{t}\int_{0}^{t}\Delta_{\pi^\star}(\tau)d\tau = 8\aaoi_{\pi^\star}. \hfil \qedhere $$
\end{proof}


\subsection{$\lim_{t\to\infty}\frac{1}{t}\sum_{\cA_n\in G_>^L(2)}(g_{m_n}-g_{0_n}) \nu_{k_n}^{\min}\le 20\aaoi_{\pi^\star}$.}

%
\begin{proof}
From \eqref{eq:r-i-n-equal}, we know that for each update $i\in\cA_n$, $r_i^L$'s are equal.
Therefore, we define  
\begin{align} \label{eq:r-A-n}
	r_{\cA_n}^L=r_{1_n}^L=r_{2_n}^L=\cdots=r_{m_n}^L.
\end{align}   
Note that $r_i^L$'s may also be equal for updates in successive subsets 
$\cA_n$ and $\cA_{n+1}$, for example, when update $k_n$ is preempted by update $k_{n+1}$. 
Therefore, we partition $G_>^L(2)$ \eqref{eq:G1-G2} into sets $\cB_1, \cB_2, \cdots$ such that $r_{\cA_n}^L=r_{\cA_{n'}}^L$ \eqref{eq:r-A-n} if and only if both $\cA_n$ and $\cA_{n'}$ lie in the same set $\cB_\ell$ (for some $\ell$). 
Further, we define $r_{\cB_\ell}^L=r_{\cA_n}^L$, $\forall \cA_n\in\cB_\ell$, and let the order $\cB_1, \cB_2, \cdots$ be such that $r_{\cB_1}^L<r_{\cB_2}^L<\cdots$. 
Then, 
\begin{align} \label{eq:split-3}
	\lim_{t\to\infty}\frac{1}{t}\sum_{\cA_n\in G_>^L(2)}(g_{m_n}-g_{0_n}) \nu_{k_n}^{\min}=\lim_{t\to\infty}\frac{1}{t}\sum_{\forall \ell}\sum_{\cA_n\in \cB_\ell}(g_{m_n}-g_{0_n}) \nu_{k_n}^{\min}.
\end{align}

Consider any particular set $\cB_\ell$. For notational clarity, we distinguish the subsets $\cA_n\in\cB_\ell$ (and related variables such as $g_{0_n}$, $g_{k_n}$, etc.) by a superscript $\ell$. Let the number of subsets $\cA_n\in\cB_\ell$ be $u_\ell$. Then, $\cA_1^\ell,\cA_2^\ell, \cdots, \cA_{u_\ell}^\ell$ denotes the subsets $\cA_n\in\cB_\ell$, in increasing order of $g_{k_n}^\ell$'s, where $g_{k_n}^\ell$ denotes the generation time of update $k_n^\ell$, and $k_n^\ell$ is the earliest generated update that SRPT$^L$ begins to transmit after all the updates in $\cA_n^\ell$ are generated. 

\begin{proposition} \label{prop:split-2}
	i) $\sum_{\cA_n^\ell\in \cB_\ell}(g_{m_n}^\ell-g_{0_n}^\ell) \nu_{k_n}^{\min,\ell}\le 10\int_{g_{k_1}^\ell}^{g_{k_1}^\ell+\nu_{k_1}^{\min,\ell}}\Delta_{\pi^\star}(\tau)d\tau$, and ii) time intervals $(g_{k_1}^\ell,g_{k_1}^\ell+\nu_{k_1}^{\min,\ell}]$ and $(g_{k_1}^{\ell+2},g_{k_1}^{\ell+2}+\nu_{k_1}^{\min,\ell+2}]$ are disjoint, $\forall \ell$.
\end{proposition}
\begin{proof}
	See Appendix \ref{app:proof-prop-split-2}.
\end{proof}

With newly defined notations, we have $\lim_{t\to\infty}\frac{1}{t}\sum_{\forall \ell}\sum_{\cA_n^\ell\in \cB_\ell}(g_{m_n}^\ell-g_{0_n}^\ell) \nu_{k_n}^{\min,\ell}$ 
\begin{align} \label{eq:split-4}
	&=\lim_{t\to\infty}\frac{1}{t}\sum_{\forall \ell\in\{1,3,5,\cdots\}}\sum_{\cA_n^\ell\in \cB_\ell}(g_{m_n}^\ell-g_{0_n}^\ell) \nu_{k_n}^{\min,\ell}+ \nonumber \\ 
	&\hspace{30ex}\lim_{t\to\infty}\frac{1}{t}\sum_{\forall \ell\in\{2,4,6,\cdots\}}\sum_{\cA_n^\ell\in \cB_\ell}(g_{m_n}^\ell-g_{0_n}^\ell) \nu_{k_n}^{\min,\ell}, \nonumber \\
	&\stackrel{(a)}{\le} \lim_{t\to\infty}\frac{1}{t}\sum_{\forall \ell\in\{1,3,5,\cdots\}}10\int_{g_{k_1}^\ell}^{g_{k_1}^\ell+\nu_{k_1}^{\min,\ell}}\Delta_{\pi^\star}(\tau)d\tau \nonumber \\
	&\hspace{30ex}+\lim_{t\to\infty}\frac{1}{t}\sum_{\forall \ell\in\{2,4,6,\cdots\}}10\int_{g_{k_1}^\ell}^{g_{k_1}^\ell+\nu_{k_1}^{\min,\ell}}\Delta_{\pi^\star}(\tau)d\tau, \nonumber \\
	&\stackrel{(b)}{\le} \lim_{t\to\infty}\frac{10}{t}\int_{0}^{t}\Delta_{\pi^\star}(\tau)d\tau
	+\lim_{t\to\infty}\frac{10}{t}\int_{0}^{t}\Delta_{\pi^\star}(\tau)d\tau, \nonumber \\
	&=20\aaoi_{\pi^\star}, 
\end{align}
where we get $(a)$ using the inequality in Proposition \ref{prop:split-2}, and $(b)$ follows because as shown in Proposition \ref{prop:split-2}, $(g_{k_1}^\ell,g_{k_1}^\ell+\nu_{k_1}^{\min,\ell}]$ and $(g_{k_1}^{\ell+2},g_{k_1}^{\ell+2}+\nu_{k_1}^{\min,\ell+2}]$ are disjoint sub-intervals of $(0,t]$, and the integrand $\Delta_{\pi^\star}(\tau)\ge 0$, $\forall \tau\ge0$. From \eqref{eq:split-3} and \eqref{eq:split-4}, we get $\lim_{t\to\infty}\frac{1}{t}\sum_{\cA_n\in G_>^L(2)}(g_{m_n}-g_{0_n}) \nu_{k_n}^{\min}\le 20\aaoi_{\pi^\star}$. 
\end{proof}

\section{Proof of Proposition \ref{prop:split-2}} \label{app:proof-prop-split-2}

\subsection{$\sum_{\cA_n^\ell\in \cB_\ell}(g_{m_n}^\ell-g_{0_n}^\ell) \nu_{k_n}^{\min,\ell}\le 10\int_{g_{k_1}^\ell}^{g_{k_1}^\ell+\nu_{k_1}^{\min,\ell}}\Delta_{\pi^\star}(\tau)d\tau$.}
\begin{proof}
For set $\cB_\ell=\{\cA_1^\ell,\cA_2^\ell,\cdots,\cA_{u_\ell}^\ell\}$, consider the sequence of updates $$0_1^\ell,k_1^\ell, 0_2^\ell,k_2^\ell, 0_3^\ell,k_3^\ell,\cdots,0_{u_\ell}^\ell,k_{u_\ell}^\ell.$$ 
By definition, $0_n^\ell$ is generated before the earliest generated update in $\cA_n^\ell$, while as discussed in Remark \ref{remark:k_n}, $k_n^\ell$ is either the latest generated update in $\cA_n^\ell$ (i.e. update $m_n^\ell$), or the earliest generated update after $m_n^\ell$. Therefore, $0_n^\ell$ and $k_n^\ell$ always refers to distinct updates, and 
\begin{align} \label{eq:k-for-m}
	g_{m_n}^\ell-g_{0_n}\le g_{k_n}^\ell-g_{0_n}^\ell.
\end{align} 
Further, since the updates in $\cA_n^\ell$ are generated before the updates in $\cA_{n+1}^\ell$ (follows from the numbering/ordering of $\cA_n^\ell$'s), we get $0_1^\ell<k_1^\ell\le 0_2^\ell<k_2^\ell\le 0_3^\ell <k_3^\ell\le\cdots\le 0_{u_\ell}^\ell<k_{u_\ell}^\ell$. Thus,
\begin{align} \label{eq:temp-simp-1}
	\sum_{\cA_n^\ell\in \cB_\ell}(g_{m_n}^\ell-g_{0_n}^\ell) \nu_{k_n}^{\min,\ell}&= (g_{m_1}^\ell-g_{0_1}^\ell) \nu_{k_1}^{\min,\ell}+\sum_{\cA_n^\ell\in \cB_\ell/\{\cA_1^\ell\}}(g_{m_n}^\ell-g_{0_n}^\ell) \nu_{k_n}^{\min,\ell}, \nonumber \\
	&\stackrel{(a)}{\le} (g_{m_1}^\ell-g_{0_1}^\ell) \nu_{k_1}^{\min,\ell}+\sum_{n\in\{2,3,4,\cdots,u_\ell\}}(g_{k_n}^\ell-g_{0_n}^\ell) \nu_{k_n}^{\min,\ell}, \nonumber \\
	&\stackrel{(b)}{\le} (\nu_{k_1}^{\min,\ell})^2+\sum_{n\in\{2,3,4,\cdots,u_\ell\}}(g_{k_n}^\ell-g_{0_n}^\ell) \nu_{k_n}^{\min,\ell},
\end{align}
where we get $(a)$ from \eqref{eq:k-for-m}, and $(b)$ follows from the fact that $\cA_1^\ell\in \cB_\ell\subseteq G_>^L(2)$, and for any subset $\cA_n^\ell\in G_>^L(2)$, $g_{m_n}^\ell-g_{0_n}^\ell<\nu_{k_n}^{\min,\ell}$ (by definition of $G_>^L(2)$).

By definition of $\cB_\ell$, the earliest time instant when an update $j\ge k_{n_\ell}^\ell$ is completely transmitted by SRPT$^L$ (i.e. $r_{k_n}^{\ell,\ell}$) is equal $\forall n\in\{1,2,3,\cdots,u_\ell\}$. Therefore, $r_{k_n}^{\ell,\ell}=r_{k_1}^{L,\ell}$, $\forall n\in\{2,3,\cdots,u_\ell\}$. 
Thus, using Lemma \ref{lemma:lb-nu-star}, we get that $\forall n\in\{2,3,\cdots,u_\ell\}$,  $r_{k_1}^{L,\ell}=g_{k_1}^{L,\ell}+\nu_{k_1}^{L,\ell}\ge g_{k_n}^\ell+\nu_{k_n}^{\min,\ell}$, which implies $(g_{k_1}^{L,\ell}-g_{k_n}^\ell)+\nu_{k_1}^{L,\ell}\ge \nu_{k_n}^{\min,\ell}$. Since $k_1^\ell<k_n^\ell$ (for $n\ge 2$), $(g_{k_1}^{L,\ell}-g_{k_n}^\ell)<0$. Therefore, $\nu_{k_1}^{L,\ell}\ge \nu_{k_n}^{\min,\ell}$.
Further, by definition of $k_1^\ell$, at some time, SRPT$^L$ begins to transmit update $k_1^\ell$. Thus, using Lemma \ref{lemma:LCFS+properties}, we get $\nu_{k_1}^{L,\ell}\le 2\nu_{k_1}^{\min,\ell}$, which implies $\nu_{k_n}^{\min,\ell}\le 2\nu_{k_1}^{\min,\ell}$.
Hence, 
\begin{align} \label{eq:simp-2}
	\sum_{n\in\{2,3,4,\cdots,u_\ell\}}(g_{k_n}^\ell-g_{0_n}^\ell) \nu_{k_n}^{\min,\ell}\le 2\nu_{k_1}^{\min,\ell}\sum_{n\in\{2,3,4,\cdots,u_\ell\}}(g_{k_n}^\ell-g_{0_n}^\ell)\stackrel{(a)}{\le}2\nu_{k_1}^{\min,\ell}(g_{k_{u_\ell}}^\ell-g_{k_1}^\ell),
\end{align}
where $(a)$ follows from the fact that $0_1^\ell<k_1^\ell\le 0_2^\ell<k_2^\ell\le 0_3^\ell <k_3^\ell\le\cdots\le 0_{u_\ell}^\ell<k_{u_\ell}^\ell$.

Again, by definition, $g_{k_{u_\ell}}^\ell\le r_{k_{u_\ell}}^\ell \stackrel{(a)}{=} r_{k_1}^{L,\ell}=g_{k_1}^\ell+\nu_{k_1}^{L,\ell}\stackrel{(b)}{\le} g_{k_1}^\ell+2\nu_{k_1}^{\min,\ell}$, where $(a)$ is true because $k_1^\ell\in\cA_1^\ell\in\cB_\ell$ and $k_{u_\ell}^\ell\in\cA_1^\ell\in\cB_\ell$ both lie in same set $\cB_\ell$, while $(b)$ follows because $\nu_{k_1}^{L,\ell}=w_{k_1}^{L,\ell}+d_{k_1}^{L,\ell}\le 2\nu_{k_1}^{\min,\ell}$ (Lemma \ref{lemma:lb-nu-star}). Thus, $(g_{k_{u_\ell}}^\ell-g_{k_1}^\ell)\le 2\nu_{k_1}^{\min,\ell}$. Substituting this result in \eqref{eq:simp-2}, we get $\sum_{n\in\{2,3,4,\cdots,u_\ell\}}(g_{k_n}^\ell-g_{0_n}^\ell) \nu_{k_n}^{\min,\ell}\le 4 (\nu_{k_1}^{\min,\ell})^2$, which together with \eqref{eq:temp-simp-1} implies 
\begin{align} \label{eq:temp-simp-3}
	\sum_{\cA_n^\ell\in \cB_\ell}(g_{m_n}^\ell-g_{0_n}^\ell) \nu_{k_n}^{\min,\ell}\le 5(\nu_{k_1}^{\min,\ell})^2. 
\end{align}

From Lemma \ref{lemma:lb-nu-star}, we get that with an optimal offline policy $\pi^\star$, the earliest time instant when the transmission of an update $j\ge k_1^\ell$ completes is $g_{k_1}^\ell+\nu_{k_1}^{\star\ell}\ge g_{k_1}^\ell+\nu_{k_1}^{\min,\ell}$. Thus, at any time $t\in(g_{k_1}^\ell,g_{k_1}^\ell+\nu_{k_1}^{\min,\ell}]$, $\lambda^\star(t)\le g_{k_1}^\ell$, which implies that the AoI of $\pi^\star$ at time $t\in(g_{k_1}^\ell,g_{k_1}^\ell+\nu_{k_1}^{\min,\ell}]$ is $\Delta_{\pi^\star}(t)= t-\lambda^\star(t)\ge t-g_{k_1}^\ell$. Hence, 
\begin{align} \label{eq:ell-interval-opt-lb}
	\int_{g_{k_1}^\ell}^{g_{k_1}^\ell+\nu_{k_1}^{\min,\ell}}\Delta_{\pi^\star}(\tau)d\tau\ge \frac{(\nu_{k_1}^{\min,\ell})^2}{2}.
\end{align}

Comparing \eqref{eq:temp-simp-3} and \eqref{eq:ell-interval-opt-lb}, we get $\sum_{\cA_n^\ell\in \cB_\ell}(g_{m_n}^\ell-g_{0_n}^\ell) \nu_{k_n}^{\min,\ell}\le 10\int_{g_{k_1}^\ell}^{g_{k_1}^\ell+\nu_{k_1}^{\min,\ell}}\Delta_{\pi^\star}(\tau)d\tau$.
\end{proof}



\subsection{Time intervals $(g_{k_1}^\ell,g_{k_1}^\ell+\nu_{k_1}^{{\min},\ell}]$ and $(g_{k_1}^{\ell+2},g_{k_1}^{\ell+2}+\nu_{k_1}^{{\min},\ell+2}]$ are disjoint, $\forall \ell$.}

\begin{proof}
By definition, $\cB_1, \cB_2,\cdots$ are such that $r_{\cB_1}^L<r_{\cB_2}^L<\cdots$, where $r_{\cB_\ell}^L=r_{k_n}^{L,\ell}$, $\forall \ell,n$ ($r_{k_n}^{L,\ell}$ denotes the earliest time instant when an update $j\ge k_n^\ell$ is completely transmitted by SRPT$^L$, while $k_n^\ell$ denotes the $k^{th}$ update in the subset $\cA_n^\ell$). Therefore, $r_{k_1}^{L,\ell}<r_{k_1}^{L,\ell+1}<r_{k_1}^{L,\ell+2}$, $\forall \ell$, which also implies that $g_{k_1}^\ell<g_{k_1}^{\ell+1}<g_{k_1}^{\ell+2}$, $\forall \ell$. Therefore, to show that the intervals $(g_{k_1}^\ell,g_{k_1}^\ell+\nu_{k_1}^{{\min},\ell}]$ and $(g_{k_1}^{\ell+2},g_{k_1}^{\ell+2}+\nu_{k_1}^{{\min},\ell+2}]$ are disjoint, we show that $g_{k_1}^\ell+\nu_{k_1}^{{\min},\ell}\le g_{k_1}^{\ell+2}$, using contradiction.

Let $g_{k_1}^\ell+\nu_{k_1}^{{\min},\ell}> g_{k_1}^{\ell+2}$. From Lemma \ref{lemma:lb-nu-star}, the earliest time by which SRPT$^L$ can completely transmit an update $j\ge k_1^\ell$ is $r_{k_1}^{L,\ell}\ge g_{k_1}^\ell+\nu_{k_1}^{{\min},\ell}$. Therefore, 
$r_{k_1}^{L,\ell}> g_{k_1}^{\ell+2}>g_{k_1}^{\ell+1}$, which implies that at time $r_{k_1}^{L,\ell}$, source has update $k_1^{\ell+2}$. 
Since SRPT$^L$ begins to transmit an update at time $t$ only if its generation time is latest, we get that at time 
$t\ge r_{k_1}^{L,\ell}$, SRPT$^L$ only transmits updates $j\ge k_2^{\ell+2}$. However, after time $r_{k_1}^{L,\ell}$, directly transmitting an update $j\ge k_2^{\ell+2}$ would imply that $r_{k_1}^{L,\ell+1}=r_{k_1}^{L,\ell+2}$, which is in contradiction to the definitions of $\cB_{\ell +1}$ and $\cB_{\ell+2}$. Therefore, the assumption that $g_{k_1}^\ell+\nu_{k_1}^{{\min},\ell}> g_{k_1}^{\ell+2}$ must be wrong. Hence, we get that $g_{k_1}^\ell+\nu_{k_1}^{{\min},\ell}\le g_{k_1}^{\ell+2}$.
\end{proof}

\end{document}